\documentclass[11pt,onecolumn]{IEEEtran}
\usepackage{amssymb}

\usepackage{color,graphicx}
\usepackage{graphicx}
\usepackage{subfigure}
\usepackage{booktabs}
\usepackage{epstopdf}
\usepackage[noend]{algpseudocode}
\usepackage{algorithmicx,algorithm}
\usepackage{dblfloatfix}
\usepackage[section]{placeins}
\usepackage{multicol}

\usepackage{lineno}
\usepackage{float}

\usepackage{amsmath}
\newtheorem{theorem}{Theorem}
\newtheorem{prop}[theorem]{Proposition}

\newtheorem{proof}{Proof}

\def\tT{{\mbox{\tiny{T}}}}
\newcommand{\bitem}{\begin{itemize}}
\newcommand{\eitem}{\end{itemize}}

\newcommand{\bpm}{\begin{pmatrix}}
\newcommand{\epm}{\end{pmatrix}}

\newcommand{\bq}{\begin{equation}}
\newcommand{\eq}{\end{equation}}

\let\abs=\envert

\let\norm=\enVert

\ifCLASSINFOpdf
\else
\fi
\begin{document}
%

\title{Variational Community Partition with Novel Network Structure Centrality Prior}
%
%
%

\author{Yiguang Bai,
        Sanyang Liu,
        Ke Yin,
        Jing Yuan
\thanks{Yiguang Bai, Sanyang Liu and Jing Yuan are with the School of Mathematics and Statistics, Xidian University, Xi'an, Shaanxi, China}
\thanks{Ke Yin is with the Center of Mathematical Sciences, Huazhong University of Science and Technology, Wuhan, Hubei, China}
}

%
%

\markboth{submitted }%
{Shell \MakeLowercase{\textit{et al.}}:  Expansion and Shrinking Benchmark Strategy for Graph Cuts in Community Detection}
%



\maketitle

\begin{abstract}
In this paper, we proposed a novel two-stage optimization method for network community partition, which is based on inherent network structure information. The introduced optimization approach utilizes the new network centrality measure of both links and vertices to construct the key affinity description of the given network, where the direct similarities between graph nodes or nodal features are not available to obtain the classical affinity matrix. Indeed, such calculated network centrality information presents the essential structure of network, hence, the proper measure for detecting network communities, which also introduces a `confidence' criterion for referencing new labeled benchmark nodes.
For the resulted challenging combinatorial optimization problem of graph clustering, the proposed optimization method iteratively employs an efficient convex optimization algorithm which is developed based under a new variational perspective of primal and dual. Experiments over both artificial  and real-world network datasets demonstrate that the proposed optimization strategy of community detection significantly improves result accuracy and outperforms the state-of-the-art algorithms in terms of accuracy and reliability.

\end{abstract}


\IEEEpeerreviewmaketitle

\section{Introduction}
\label{S:1}

Modern network science~\cite{Guanrong2014Network,Watts1998Collectivedynamics,Barabasi1999Albert,mej2010networks,boccaletti2006complex} has brought
crucial and significant improvements to our understanding of complex system~\cite{mitchell2006complex}. One of the most prominent features for representing real complex network systems is the structure of communities~\cite{goldenberg2001talk,palla2005uncovering}, i.e.
the organization of vertices, for which vertices in the same community have more connections than ones in the other  communities~\cite{fortunato2010community}.
In a complex system, each community is often composed of multiple entities with similar properties, which provides a deep insight into the structure and function of the whole network system~\cite{yang2015unified,rosvall2007information}. Therefore, detecting communities is of great importance for many different applications of biology,
physics, sociology and computer science, where the system is usually modeled as a complicated network with edges linking each related node pairs  \cite{newman2004detecting}.

With this respect, many different methods were proposed to solve the challenging problem of partitioning independent communities from the given network: one of the most popular methods, proposed by Girvan and Newman \cite{girvan2002community}, is to identify network communities through maximizing the modularity $Q$ of the associate network, which has been the essential criterion of many community detection methods. However,
optimizing a network's modularity is mathematically nontrivial. Recent studies demonstrated that exactly maximizing the modularity of network is NP-complete, for which many polynomial-time approximation methods have been proposed, such as the greedy method \cite{newman2004detecting},
simulated annealing \cite{guimera2004modularity}, extremal optimization \cite{duch2005community}, intelligent
optimization \cite{liu2014multiobjective,ma2014multi}, game theory-based methods \cite{Bu2017Dynamic} and spectral methods \cite{ali2017improved}.
Another way to deal with such a network partition problem is to construct the affinity matrix of the associate network graph, where the affinity matrix encodes the seminal information about network structure and connectivities, and compute the optimized nodal labels over the given graph which directly separate the network graph into different partitions or communities. According to the proposed optimization criterion of labeling estimation, many different approaches were proposed to recover such labels of graph nodes, for example, label propagation and random walks over graphs \cite{Zhu03semi-supervisedlearning,Azran07therendezvous}, spectral clustering \cite{Chung1996Spectral,Luxburg2007A} based balanced graph cut approaches including ratio cut \cite{Hagen1991Fast}, normalized cut \cite{Shi2000Normalized}, p-Laplacian based graph cut \cite{buhler2009spectral} and cheeger cut \cite{hein2011beyond} etc., however, most of them are computationally expensive and inefficiency due to their posed non-convex optimization formulations. In contrast, min-cut-based methods \cite{Blum01learningfrom,Boykov01fastapproximate} can solve
the studied combinatorial optimization problem efficiently in an approximate way along with the capability to handle large-scale graphs. Recently, convex optimization was developed as a popular optimization framework to build up fast solvers for labeling recovery in the spatially continuous
setting \cite{MR2177726,bresson2014multi,Yuan2010A,yuan2010continuous} etc.: the main ideas of related convex optimization approaches, i.e. relax the binary label constraints to a continuous convex set and round the result of the reduced convex optimization problem back to binary, can be directly
extended to discover the optimum labeling of each graph node, for example Yin et al's total-variation-based region force method (TVRF) to semi-supervised clustering \cite{yin2018effective}, which introduced a fast splitting optimization framework to the proposed convex optimization problem and outperformed most state-of-the-art graph partitioning algorithms.

Besides the matter of developing efficient optimization algorithms, another major challenges of detecting communities from some network are in two folds: First, it lacks reliable descriptions to encode inherent network structures and affiliation of any node to the specified community. One common way is to sample some graph nodes into some communities beforehand, so called benchmark nodes, which can partially solve this difficulty: benchmark nodes reveal limited network structure information which can still indicate the recovery of other nodes belonging to the same community through evaluating the given network's coreness, betweenness etc.; in addition, benchmark nodes provide meaningful starting points to propagate label along graph links to the other nodes, i.e. label inference, which can be incorporated into the often-used optimization procedures. Second, it could be hard to get enough benchmark nodes in real-world applications; especially for the large sparse networks, low degree benchmark nodes can hardly provide useful information about network communities and worsen performance of the followed partition procedure. Therefore, how to discover more 'trustable' benchmark nodes during the whole computing process becomes one key factor of extracting partitions or communities accurately from the given network.

\subsection{Contributions and Organization}
Motivated by previous studies, we present a novel two-stage optimization method for network community partition,
based on the internal structure measure of the given network graph. The proposed optimization approach utilizes a new network centrality measure of both links and vertices to construct the key affinity matrix of the given network correctly, for the cases vanishing node similarities
which commonly happen for network community detection.
The network centrality information actually reveals the essential structure of network, which, hence, provides a proper clue for detecting network communities and introduces an additional `confidence' criterion for labelings by referencing the labeled benchmark nodes.
Particularly, the two-stage optimization method makes use of the network centrality-based `confidence' measure for the stage of benchmark node refinement, and an efficient convex optimization algorithm to the solve the followed challenging combinatorial optimization problem of graph clustering, which is developed based under a new variational perspective of primal and dual. Refining benchmark nodes can effectively improve the accuracy and reliability of the proposed optimization approach. Experiments over both artificial and real-world network datasets demonstrate that the proposed optimization method of community detection outperforms the state-of-art algorithms in terms of accuracy and reliability.

\subsection{Definitions and Notations}
Let $\mathbb{C}:=\{ {C_1},{C_2}, \cdots {C_K}\}$ be the set of $K$ communities, which is represented by the graph $G:=(V,E)$
with $\abs{V} = n$ vertices (or nodes) and
$\left| E \right| = e$ edges (or links);  each edge $e_{ij}$, where $i,j \in \{1 \ldots n\}$, denotes the existing link between two nodes $v_i$ and $v_j$,
and each community $C_k=(V_k,E_k)$, $k=1 \ldots K$, is a
distinct subgraph of $G$.
The connectivity of $G$ can be expressed as its adjacency matrix $A = (a_{i j})$
whose $(i,j)$-entry $a_{i j} = 1$ means there exists a link between the two nodes $v_i$ and $v_j$, and $a_{i j}=0$ otherwise.
The matrix $W = (w_{i j})$ represents the affinity matrix of graph $G$, where $w_{{i j}}$ measures the similarity between the two vertices
of $v_i$ and $v_j$,
and is usually given as a symmetric matrix with non-negative entries. Additionally, the diagonal matrix $D=(d_{ii})$ is given by
$d_{ii} = \sum_{j=1}^n w_{ij}$, $i=1 ... n$.

With this, the linear operators of gradient and divergence over the graph $G$ are introduced as follows~\cite{Zhou2005Regularization}:
for some scalar function $u(v_i)$ given at each node $v_i$, its gradient  $\nabla_{e_{ij}} u$ evaluates
the difference of $u(\cdot)$ between two nodes $v_i$ and $v_j$ along the link $e_{ij}$ such that

\begin{equation}
  \nabla_{e_{ij}} u \, = \, w_{ij}(u({v_j}) - u({v_i}))_{{v_j} \in \mathcal{N}({v_i})} \, ,
  \label{func_gradian}
\end{equation}
whose $L_p$-norm, $p>= 1$, is measured as
\bq  \label{eq:grad-norm}
\norm{\nabla u}_p \, = \, \sum\limits_{e_{i j} \in E} w_{ij}\, \abs{u({v_j}) - u({v_i})}^p \, ;
\eq

for some function $f(e_{ij})$ given at each edge $e_{ij} \in E$, its divergence $\mathrm{div}(f)_i$ at the node $v_i$ measures the balance of
$f$ over all the edges associate linking the neighbour nodes around $v_i$, i.e. $\mathcal{N}(v_i)$ such that
\begin{equation}\label{eq:div}
\mathrm{div}(f)_i\, = \, \sum_{{v_j} \in \mathcal{N}(v_i)} f(e_{ij})\, .
\end{equation}

\section{Semi-Supervised Graph Partition and Convex Optimization Model}
\label{S:2}

In this work, we aim to partition communities from a given graph network with a new two-stage optimization method. The proposed optimization approach utilizes a new network centrality measure of both links and vertices to construct the key affinity matrix of the given network correctly, for the cases vanishing node similarities
which commonly happen for network community detection.
The network centrality information actually reveals the essential structure of network, which, hence, provides a proper clue for detecting network communities and introduces an additional `confidence' criterion for labelings by referencing the labeled benchmark nodes.
Particularly, the two-stage optimization method makes use of the network centrality-based `confidence' measure for the stage of benchmark node refinement, and an efficient convex optimization algorithm to the solve the followed challenging combinatorial optimization problem of graph clustering, which is developed based under a new variational perspective of primal and dual. Refining benchmark nodes can effectively improve the accuracy and reliability of the proposed optimization approach.

\subsection{Semi-Supervised Graph Partitioning}

Graph partitioning targets to cut the given graph $G$ into multiple independent subgraphs (or communities). Let ${\bf \Psi}  = (\psi _{i k})$ be a binary matrix, where $\psi_{i k}=\{0,1\}$ denotes the node $v_i$ belongs to the community $C_k$  $(\psi_{i k}=1)$ or not $(\psi_{i k}=0)$. Then, graph partitioning tries to minimize the following energy function:
\begin{equation}
 E = \sum_{k = 1}^K \sum_{e_{i j} \in E} w_{ij}\abs{\psi _{ik} - \psi _{jk}} \, , \quad i=1\, \ldots \, n\, , \; k=1\, \ldots \, K \, .  \label{func_1}
\end{equation}
Also, the convex penalty function of each term in \eqref{func_1} can also be the quadratic function $\abs{\cdot}^2$ or $\ell_p$-norm function $\abs{\cdot}^p$, $p>1$, which results in the weighted Laplacian or $p$-Laplacian as the energy function of \eqref{func_1} \cite{Luxburg2007A,buhler2009spectral}.

In addition, each vertex belongs to only one subgraph/community, i.e.
\begin{equation} \label{eq:const1}
\sum_{k=1}^K \, \psi_{ik} \, = \, 1\, , \quad i = 1 \, \ldots \, n \, .
\end{equation}

Using the definitions of graph gradients \eqref{func_gradian} and \eqref{eq:grad-norm}, the optimization problem \eqref{func_1} can be written in a more concise form
of minimizing the corresponding graph total-variation function such that
\begin{equation}
\min_{\psi_{i k} \in \{0,1\}} \; \sum\limits_{k = 1}^K \, \norm{\nabla \Psi_k} _1 \,, \quad \text{s.t.  \eqref{eq:const1}} \, ;
\label{func_3}
\end{equation}
where $\Psi_k=(\psi_{1,k}, \ldots, \psi_{n k})^\tT$ denotes the $k$-th column of ${\bf \Psi}$.

It is clear that the optimization model \eqref{eq:potts2} has a trivial solution, where all vertices belong to the same community.
One important way to avoid this situation is to integrate priori information into the proposed optimization problem \eqref{func_1}, for example, some benchmark nodes are labeled, hence separating such partially labeled graph into multiple independent subgraphs/communities introduces a proper \emph{semi-supervised graph partition problem} \cite{Book2009Zhu,yin2018effective}. Indeed,
the pre-labeled nodes helps improving the partition results for the given graph mainly in two folds: first, the labeled nodes provide the starting positions to propagate labels to the other vertices \cite{Book2009Zhu}; second, they also
reveal the essential features to construct graph partitioning hints in geometry or other respects (see the following section).

Let $S_k \subset C_k$, $k=1 \ldots K$, be the benchmark set which represents a sample fraction of the community $k$, and the total benchmark set $S =  \cup _{k = 1}^K{S_k}$. To this end, we have
\bq
\forall i \in S_k\, , \quad \psi_{i j} \, = \, \left\{ \begin{array}{ll}1 \, & \text{if  } j =  k \\ 0 \, & \text{if } j \neq k \end{array} \right. \, ,
\quad k\, =\, 1 \ldots K \, .
\eq

With the locations of pre-labeled nodes, we define the novel measure $p_{ik}$, $i=1 ... n$ and $k=1 ... K$, which characterizes the probability of each vertex $v_i$ belonging to community $C_k$ such that
\begin{equation}
p_{ik} \, = \, \frac{{\frac{1}{{\left| {{S_k}} \right|}}\sum\limits_{j \in {S_k}} {{q_{ij}}} }}{{\sum\limits_{r = 1}^K {\frac{1}{{\left| {{S_r}} \right|}}\sum\limits_{j \in {S_r}} {{q_{ij}}} } }} \, , \quad
\text{where}\;\;
{q_{ij}} = \frac{{{{({{\widehat w}_{ij}})}^2}}}{{{{\widehat w}_{ii}}{{\widehat w}_{jj}}}}
\label{similarityD}
\end{equation}
and  the matrix $({\widehat w_{ij}}) = {D^{ - 1/2}}W{D^{ - 1/2}}$ is the corresponding normalized affinity matrix; when the denominator is zero, set ${p_{ik}} = \frac{1}{K}$.
Therefore, we can integrate the cross-entropy information between the possibility $p_{ik}$ and the label function $\psi_{ik}$ into the optimization model \eqref{func_3} which gives rise to the following optimization problem:
\begin{equation}
\min_{\psi _{ik} \in \{0,1\}} (1 - \tau ) \sum_{k = 1}^K \, \norm{\nabla \Psi_k} _1\, + \, \tau \sum_{k = 1}^K \, \sum_{i = 1}^n\, \Big(- \log(p_{ik})(1 - \psi _{ik}) \, - \,\log(1 - p_{ik} )\psi_{ik}\Big)
\\label{fi_function}
\end{equation}
subject to the constraint \eqref{eq:const1}.

\subsection{Convex Relaxation and Dual Optimization}

Finding the optimum $\mathbf{\Psi}$ to the proposed minimization problem \eqref{fi_function} over the binary constraint $\psi_{ik} \in \{0,1\}$ is challenging, actually NP hard which means there is no efficient polynomial-time algorithm for such combinatorial optimization problem \eqref{fi_function}. In practice, we often replace the binary constraint $\psi_{i k} \in \{0, 1\}$ by its convex relaxation $\psi_{i k} \in [0,1]$ instead; hence, we have
\begin{equation}
\min_{\psi _{ik} \in [0,1]} (1 - \tau ) \sum_{k = 1}^K \, \norm{\nabla \Psi_k} _1\, + \, \tau \sum_{k = 1}^K \, \sum_{i = 1}^n\, \Big(- \log(p_{ik})(1 - \psi _{ik}) \, - \,\log(1 - p_{ik} )\psi_{ik}\Big) \,, \quad \text{s.t.  \eqref{eq:const1}} \, .
\label{eq:potts}
\end{equation}

On the other hand, combine the two constraints $\psi_{i k} \in [0,1]$ and \eqref{eq:const1}, i.e.
\bq \label{eq:simp}
\psi_{i k} \, \geq \, 0 \, , \quad \; \sum_{k=1}^K \, \psi_{ik} \, = \, 1\, , \quad i = 1 \, \ldots \, n \, ,
\eq
which denotes that, for each node $v_i$, the $i$-th row $\Psi^i= (\psi_{i 1} , \ldots , \psi_{i K})$ of the matrix ${\bf \Psi}$ belongs to
the $K$-dim simplex set $T_K$.

In this sense, we can rewrite the convex optimization problem \eqref{eq:potts}, also in view of \eqref{func_1}, as
\bq \label{eq:potts2}
\min_{\psi} (1 - \tau ) \sum_{k = 1}^K \sum_{e_{i j} \in E} w_{ij}\abs{\psi _{ik} - \psi _{jk}} \, + \, \tau \sum_{k = 1}^K \, \sum_{i = 1}^n\, M_{i k} \psi _{ik} \,,
\eq
where $M_{i k} = \log (p_{ik} / (1 - p_{ik}))$, subject to
\[
\Psi^i \, \in \, T_K \, , \; i = 1 \ldots n\, .
\]

Through variational analysis, see appendix \ref{sec:app1} for details, we can prove the equivalence between the convex optimization problem \eqref{eq:potts2} and its associate dual model \eqref{eq:app-dual} such that
\begin{prop} \label{prop1}
The convex optimization problem \eqref{eq:potts} is mathematically equal to the following maximization problem
\bq \label{eq:dual}
\max_{q, r}\; \sum_{i=1}^n r_i^s \, ,
\eq
subject to
\bq
\mathrm{div}(q^k)_i \, - \, r_i^s \, + \, r_i^k \, = \, 0 \,, \;\;
\abs{q^k(e_{ij})} \, \leq \, w_{ij} \, , \;\;  r_i^k \, \leq \, M_{ik}\, , \quad i \, = \, 1 \ldots n\, , \;\; k \, = \, 1 \ldots K \, .
\eq

In addition, the optimum $\psi_{ik}$, $i=1 ... n$ and $k=1 ... K$, to the original convex optimization problem \eqref{eq:potts2} are just the optimal multipliers to the above linear equality constraints.
\end{prop}
\begin{proof}
The proof can be found in the appendix \ref{sec:app1}.
\end{proof}

Given the fact that the dual optimization model \eqref{eq:dual} is equivalent to the studied convex optimization problem \eqref{eq:potts2}, where
the optimum ${\bf \Psi}$ to \eqref{eq:potts2} works as the optimal multipliers to the linear equality of the dual problem \eqref{eq:dual}, the energy function of the respective primal-dual model \eqref{eq:pd} is just the conventional Lagrangian function of \eqref{eq:dual}:
\[
L(\psi, q, r) \, = \, \sum_{i=1}^n r_i^s \, + \, \sum_{k=1}^K \sum_{i=1}^n \psi_{ik}
\Big( \mathrm{div}(q^k)_i \, - \, r_i^s \, + \, r_i^k \Big)\, .
\]

In this paper, we employ the classical augmented Lagrangian method (ALM) \cite{citeulike1859441} to construct a novel efficient ALM-based algorithm to tackle the linear equality constrained dual optimization problem \eqref{eq:dual}, which can resolve both $\psi$ and the additional dual variables $(q,r)$ simultaneously.
Upon the above classical Lagrangian function, we define its augmented Lagrangian function
\[
L_c(\psi, q, r) \, = \, L(\psi, q, r) \, - \, \frac{c}{2} \sum_{k=1}^K \sum_{i=1}^n \Big( \mathrm{div}(q^k)_i \, - \, r_i^s \, + \, r_i^k \Big)^2\, .
\]

Therefore, the proposed ALM-based algorithm to the dual optimization problem \eqref{eq:dual} explores two major steps at each iteration till convergence (see Alg. \ref{alg-1} in the appendix \ref{sec:app2} for details):
\begin{enumerate}
\item fix $\psi^{t-1}$, compute $(q, r)^t$ by maximizing $L_c(\psi, q, r)$:
\[
(q, r)^t \, = \, \arg \max_{q, r} \, L(\psi^{t-1}, q, r) \, ;
\]
\item update $\psi^t$ by the computed $(q, r)^t$:
\[
\psi^t \, = \, \psi^{t-1} \, - \, c \big(\mathrm{div}(q^k) \, - \, r^s \, + \, r^k\big)^t\, .
\]
\end{enumerate}

Once the proposed ALM-based (Alg. \ref{alg-1}) converges to some optimum $(\psi_{ik}^*)$, $i=1 ... n$ and $k=1 ... K$, we can simply round $(\psi_{ik}^*)$ into its binary version, such that for each node $v_i$, $\psi^*_{ik} = 1$ when $k = \arg \max (\psi^*_{i1}, ..., \psi^*_{iK})$ and $\psi^*_{i j\neq k} = 0$.

\section{Network Structure Centralities and Two-Stage Community Partition}

Clearly, it is the key factor for partitioning a graph or network accurately that the right affinity description $(w_{i j})$ is provided for \eqref{fi_function} and \eqref{similarityD}. The classical way for most state-of-art clustering methods is to employ similarities between nodes or specified nodal features for constructing the associate affinity matrix, which is, however, unavailable in many cases of network clustering. In this work, we propose a novel method to calculate such network affinity matrix $(w_{ij})$ based upon inherent structure centrality of the network links and vertices, and introduce a new two-stage optimization strategy (TSOS) to cluster the communities with both efficiency and accuracy.

\subsection{Network Structure and Betweenness of Links}

In this section, we define the affinity/adjacency matrix $(w_{ij})$ directly from the network structure information of link centrality, i.e. betweenness of network links.

In fact, betweenness of the network link $e_{ij}$ is defined as the total number of shortest paths that pass through $e_{ij}$ \cite{bai2017effective,Dunn2005The} from all vertices to all the other vertices, such that
\begin{equation}
\mathrm{BCL}(e_{ij})\,=\,\sum\limits_{l\ne k}{\frac{{{g}_{lk}}(e_{ij})}{{{g}_{lk}}}} \, ,
\end{equation}
where $g_{lk}$ is the total number of all shortest paths from any node $v_l$ to a different node $v_k$, and $g_{lk}(e_{ij})$ is the number of such paths through the link $e_{ij}$.

\begin{figure}[h]
\centering
  \includegraphics[width=0.5\textwidth]{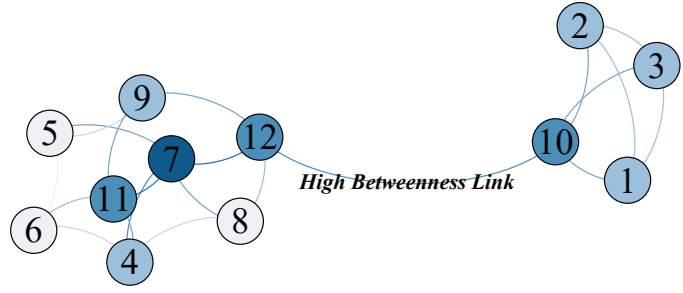}
\caption{High betweenness of a link can be viewed as the bridge to connect two communities.
As shown above, the link between the node $v_{12}$ and the node $v_{10}$ can be viewed as the key bridge edge between two distinct communities.}
\label{lin}
\end{figure}

Betweenness of the link is actually one of the most important factors for network partition: high betweenness of a link can be taken as the bridge to connect two communities, which means that once removed, the number of isolated network blocks would increase\cite{Bollob1998Modern}. Indeed, this is exactly the expected edge to separate the network.
To this end, for an edge $e_{ij}$, we can define the corresponding weight $w_{ij} = f(\mathrm{BCL}(e_{ij}))$ where $f(\cdot)$ is a positive strictly decrease function, i.e. the cost of cutting the edge $e_{ij}$ with high betweenness value is low, hence partitioning would likely happen on this edge.
In this paper, we consider the inverse of betweenness $\mathrm{BCL}$ as the definition of the affinity matrix $(w_{ij})$:
\begin{equation} \label{equation:similarity matrix}
{w_{ij}} = \left\{ \begin{array}{ll}
1/\mathrm{BCL}(e_{ij}) & \quad \text{if}\;  e_{ij}\, \in\, E\\
0 & \quad  \text{otherwise}
\end{array} \right. \, .
\end{equation}
Hence, the parameter values $p_{ik}$ of the optimization problem \eqref{eq:potts} can be computed through \eqref{similarityD}.

\subsection{Nodal Centrality, Benchmark Confidence and Two-Stage Optimization Strategy (TSOS)}

\begin{figure}[h]
\centering
  \includegraphics[width=0.38\textwidth]{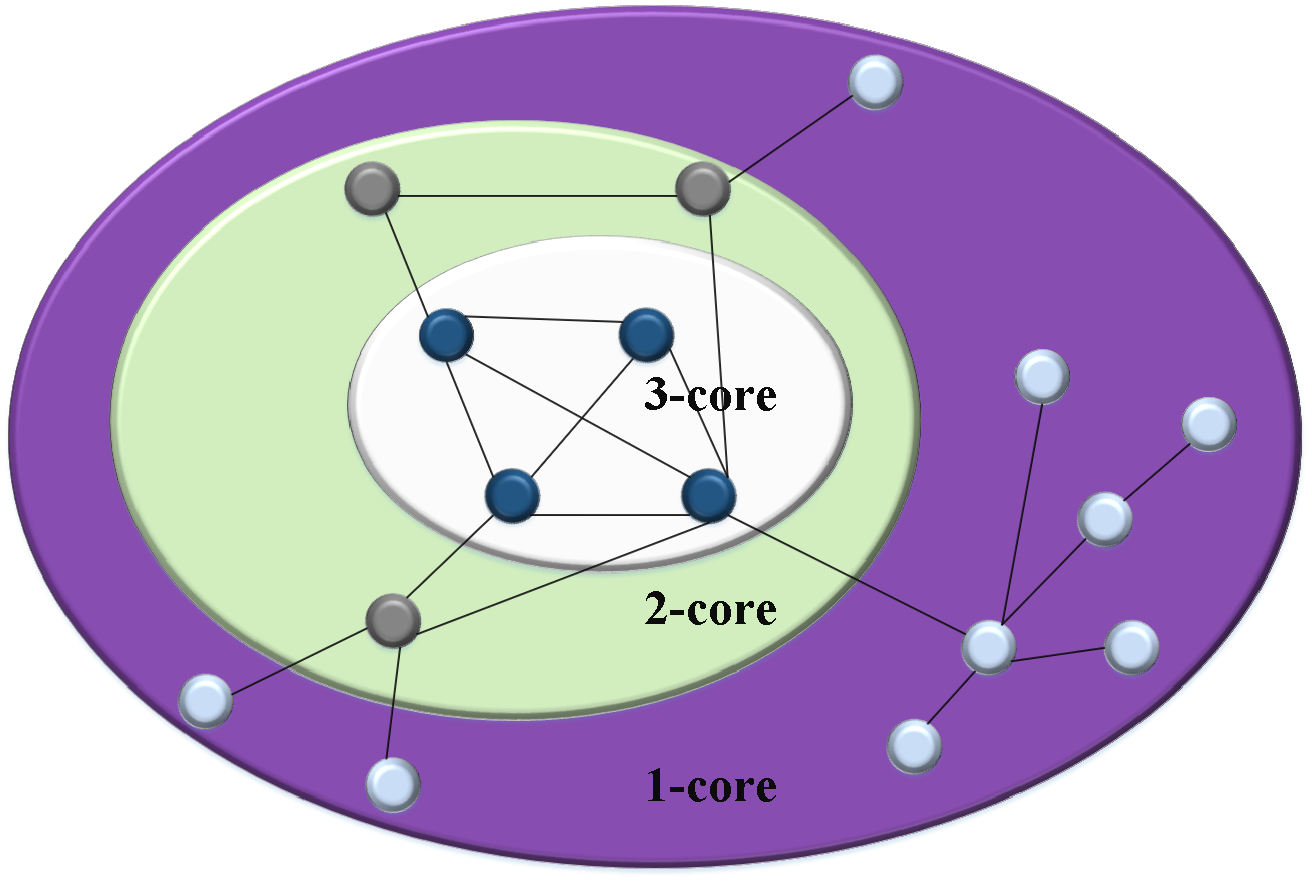}
  \caption{The diagram of k-shell decomposition, including 1-core nodes (violet), 2-core nodes (light green) and 3-core nodes (white).
}
\label{k-shell}
\end{figure}

Actually, the `core members' or `core nodes' of each network community are closely connected with each other and dominate more nodes than the other ones within one hop range, which defines the centrality of network nodes. Clearly, such core nodes with the correct label will directly find the other core nodes with the same label, i.e. the core members of the respective community, once the core nodes are discovered beforehand (see the following section for details).

The topological centrality of each node can be quantified through the concept of $k$-core, which is defined as the largest subnetwork in which every node has at least $k$ links, i.e. with the degree $k$. As shown in Fig. \ref{k-shell}, the k-core of a given network can be obtained by recursively removing all nodes with the degree less than $k$, until all the nodes in the remaining network have the degree not less than $k$. Repeating this for $k =1, 2, ...$, finally determines the $k$-shell decomposition of a network.
Hence, the coreness of each node $v_i$, $i=1 ... n$, is then defined as the integer $\gamma_i$ for which this node belongs to the $\gamma_i$-core but not to the $(\gamma_i + 1)$-core~\cite{yang2017small,dorogovtsev2006k}.
In general,  the node with a bigger coreness value must have a higher centrality. In this paper, we therefore adopt such coreness to characterize the nodal centrality and the `core nodes', or `core members', are the ones with the highest coreness number.

On the other hand, the coreness of any node defines the closeness of such node to the `core nodes' of the associate community; hence, we
define the related `confidence' measure that evaluates the possibility of any node $v_i$ belonging to the corresponding community $k$:
\bq \label{eq:confidence}
{\pi _{ik}} \, = \, {\gamma _i} \cdot  \max_{j \in {S_k}}w_{ij} \, , \quad \, i \, = \,1 \cdots n\, 
\eq
where bigger $w_{ij}$, $j \in S_k$, means lower betweenness in terms of \eqref{equation:similarity matrix} and lower likelihood for cutting the associate link $e_{ij}$, so higher `confidence' to associate two nodes $v_i$ and $v_j$. In this sense, $\pi_{ik}$ actually evaluates the `confidence'
to combine the node $v_i$ into the benchmark set of the community $k$.

With such `confidence' measure, we introduce a two-stage optimization framework: it first computes an initial network partition through the proposed
ALM-based dual optimization algorithm; once the initial partitioned communities are obtained, the `confidence' measure $\pi_{ik}$ for each node $v_i$ within its initial partition $k$ is calculated by \eqref{eq:confidence}, so as to choose the new nodes with high `confidence' as \eqref{eq:sigma} into the related benchmark set $k$; using the expanded benchmark sets, the proposed ALM-based dual optimization is explored to recompute the network partition. More details of the two-stage optimization strategy can be found in Alg. \ref{alg-2}.

In fact, the proposed two-stage optimization strategy does not require many initial benchmark nodes to ensure the accuracy of network partition, since the benchmark sets can be expanded with more dominate nodes of high confidence.
Meanwhile, the two-stage optimization method, along with increasing benchmark nodes, essentially reduces the total number of undetermined graph nodes, this improves efficiency of the following partition procedure. On the other hand, the alternating steps of optimization and benchmark expansion can be performed not only two but also more than two times, hence a multi-stage optimization method. In practice, we found the two-stage-optimization can reach the result good enough, using more than two optimization stages does not improve the results significantly (see the experiment results of Fig. \ref{exp:ssl} for details).

In this work, we often pick nodes with high `confidence' into the benchmark set, whose related $\pi_{ik}$ suffice the following condition:
\bq \label{eq:sigma}
\pi_{ik} \, \geq \, \bar{\pi}_k \, +\, \delta \sigma_k
\eq
where $\bar{\pi}_k$ and $\sigma_k$ are the average and standard deviation of all the values $\pi_{ik}$, and $\delta > 0$, see Sec. \ref{sec:delta} for choosing the proper parameter $\delta$ for experiments in this work.

\begin{algorithm}[!h]
    \caption{Two-Stage Optimization Strategy \label{alg-2} }
    \begin{algorithmic}[1]
    \State {\bf Setup up:} choose benchmark nodes $S_{k} ~~(k = 1,2, \cdots ,K)$, calculate the affinity matrix $ (w_{ij})$ and the costs $M_{ik}$, $i=1 ... n$ and $k=1 ... K$, by \eqref{similarityD} and \eqref{equation:similarity matrix};
    \State {\bf Compute initial partitions:} utilize the proposed ALM-based dual optimization algorithm (Alg. \ref{alg-1}) to compute the initial partition results $\psi_{ik}^*$, $i=1 ... n$ and $k=1 ... K$;
    \State {\bf Expansion of benchmark sets:} with the initial partition results, the 'confidence' measure $\pi_{ik}$ for each node $v_i$ within its initial partition $k$ is calculated by \eqref{eq:confidence}, choose the new nodes with high `confidence', e.g. \eqref{eq:sigma}, into the related benchmark set $k$;
    \State {\bf Refine partitions:} use the expanded benchmark sets, the proposed ALM-based dual optimization algorithm (Alg. \ref{alg-1}) is employed to recompute the network partition.
    \end{algorithmic}
\end{algorithm}

\section{Experiments}
We, in this work, explore two artificial networks of GN and LFR and 5 real-world networks to validate the effectiveness and efficiency of the proposed
two-stage optimization strategy (TSOS), see Alg. \ref{alg-2}, for partitioning  the given network into multiple communities with inherent network structure information.
Experiment results are recorded from the average performance of 20 independent trials, and compared with ground-truth.

For the unweighted networks of GN, LFR, Dophin, Football, and Polbooks, their affinity matrices $(w_{ij})$ are calculated by \eqref{equation:similarity matrix}. For the data clustering networks of COIL and MINST, we adopt their given similarity weights to construct their affinity matrices $(w_{ij})$ directly by equation \eqref{similarityD}. In addition, we compare our proposed method with one of state-of-the-art data clustering approach
proposed by Yin et al \cite{yin2018effective}, namely the total-variation-based data clustering algorithm with region force (TVRF).

\subsection{Experiments of Benchmark Expansion, Optimization Stages and Parameter $\delta$}
\label{sec:delta}
\begin{figure}[h!]
\centering
\begin{tabular}[p]{c@{\,}c}
\subfigure[]{
  \includegraphics[width=0.46\textwidth]{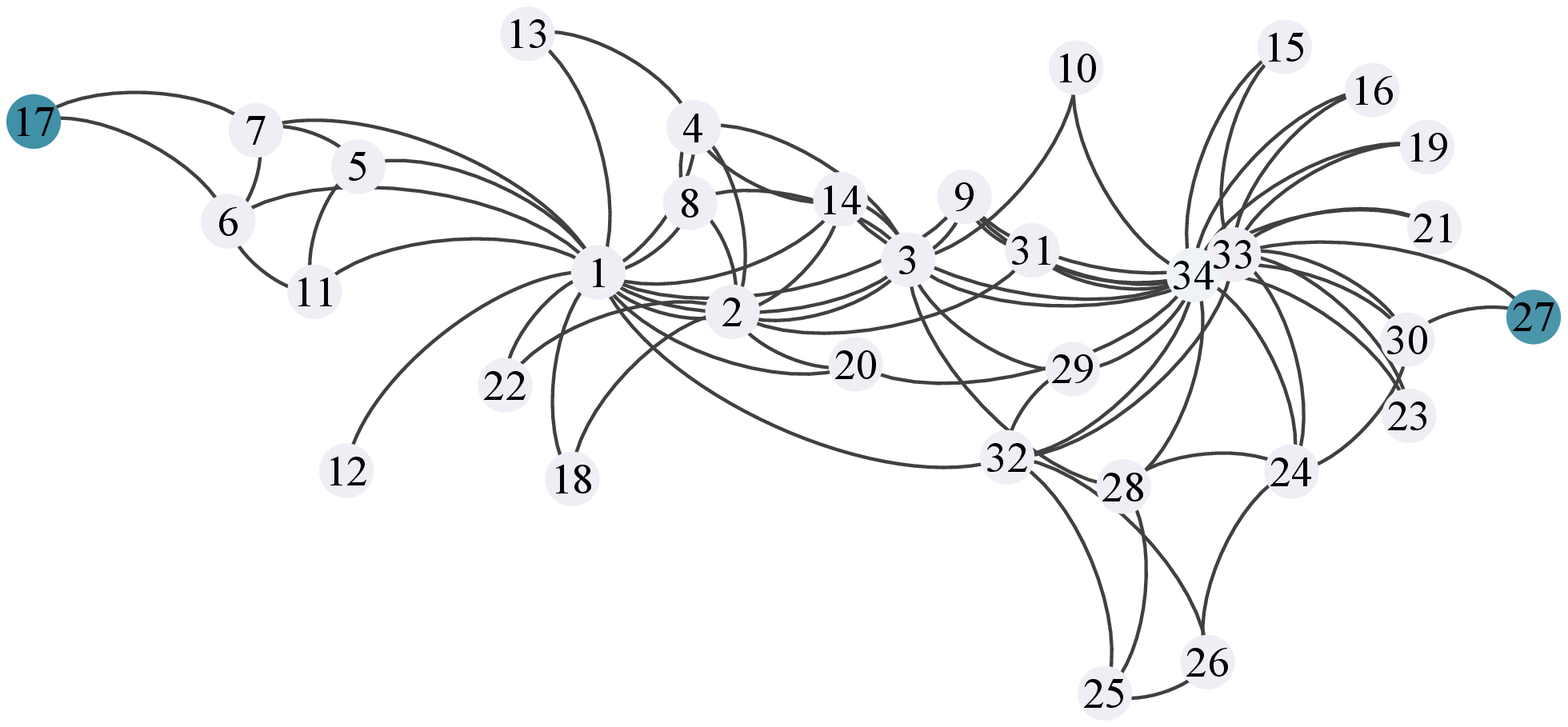}}&
  \subfigure[]{
  \includegraphics[width=0.46\textwidth]{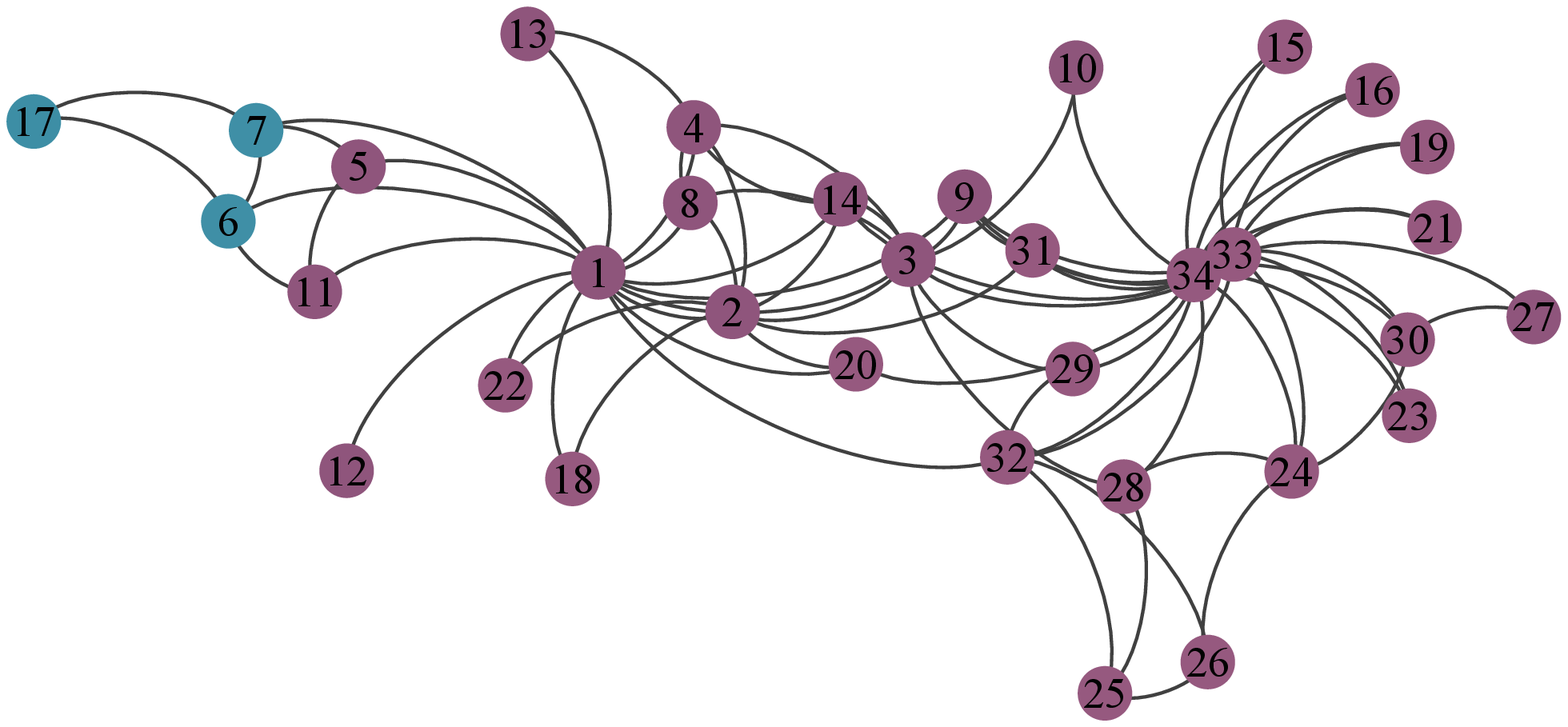}} \\
  \subfigure[]{
  \includegraphics[width=0.46\textwidth]{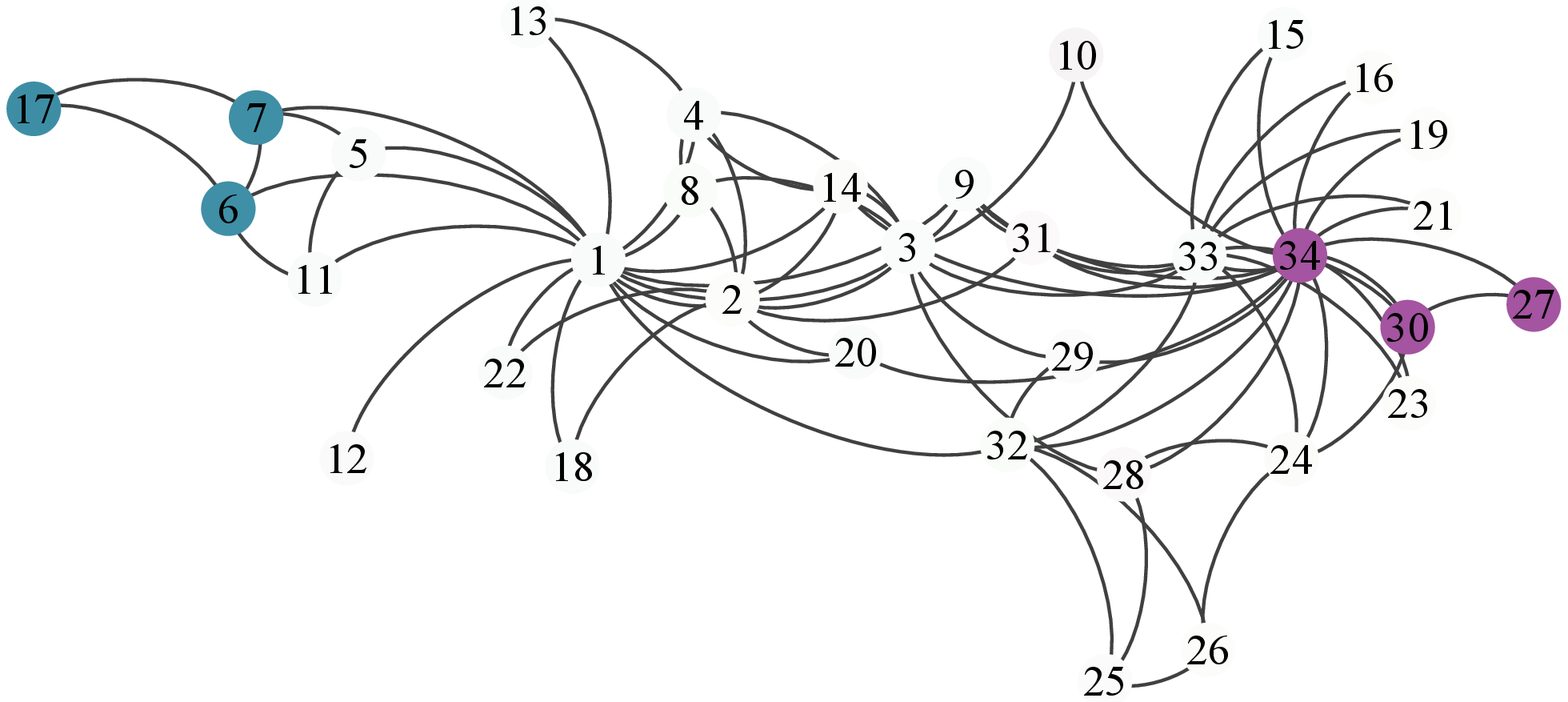}}&
  \subfigure[]{
  \includegraphics[width=0.46\textwidth]{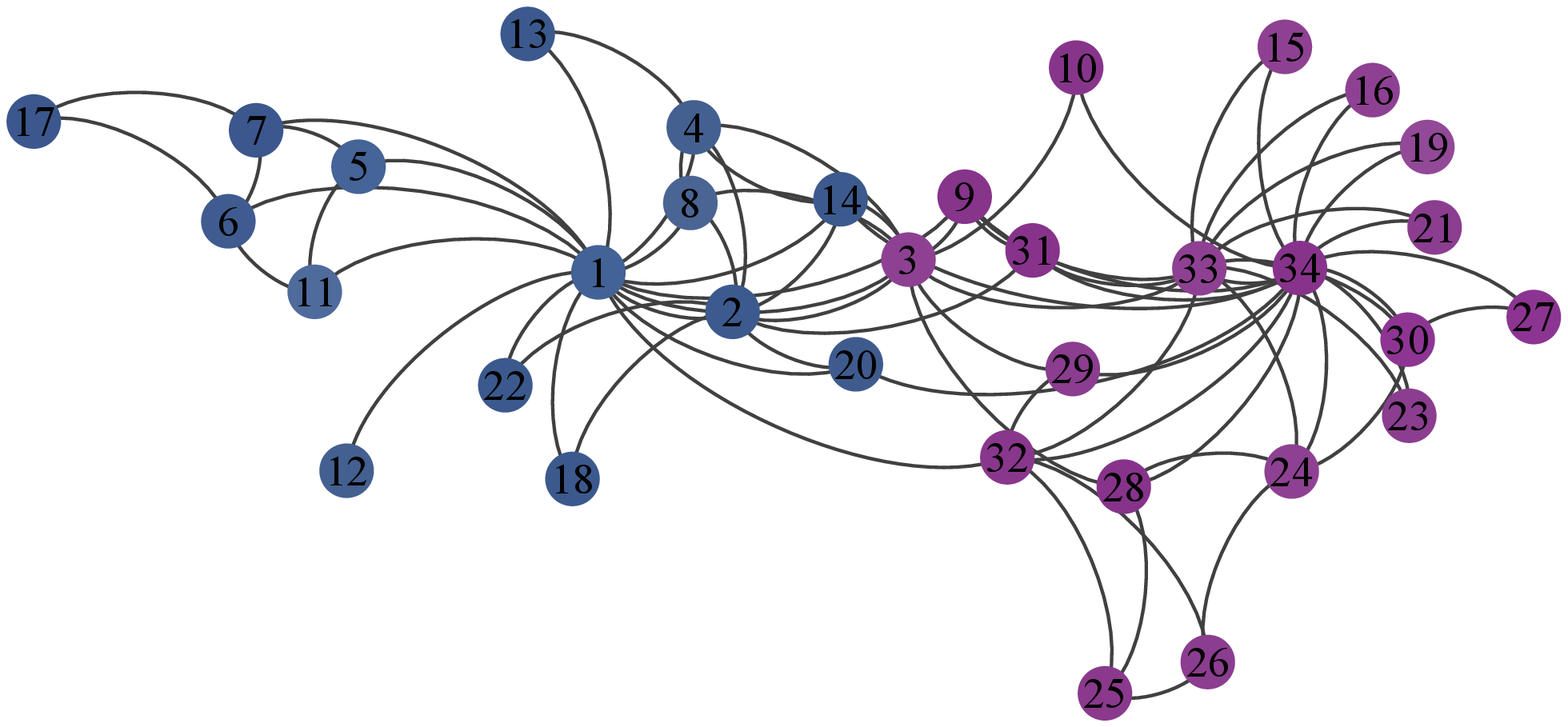}}
  \end{tabular}
\caption{(a). Nodes $v_{17}$ and $v_{27}$ are picked as the benchmark nodes for the two network communities respectively.
(b). Accuracy of the initial partition result by the proposed ALM-based algorithm is only $61.77\%$, due to its sparse network structure for which benchmark nodes can not provide more information to help partition.
(c). Expansion of benchmark nodes are performed by \eqref{eq:sigma} with the benchmark confidence measure \eqref{eq:confidence}.
(d). The followed partition procedure using ALM-based algorithm largely improves the accuracy to $97.06\%$.
\label{Toy}
}
\end{figure}

\subsubsection{Bechmark Expansion}
In this section, we show the proposed two-stage optimization strategy (TSOS) significantly improve the community partition results of networks, especially the sparsely connected networks. For the given sparse network of Zachary karate clubs consisting 34 vertices and 2 communities, as illustrated in Fig. \ref{Toy}, the labeled benchmark nodes dominate very few nodes (see Fig. \ref{Toy}(a)), thus provides not much network structure information and results in inaccurate partition result initially (see Fig. \ref{Toy} (b)).
Actually, shortage of pre-labeled nodes, or benchmark nodes with sufficient dominates, is often the big challenge for the state-of-the-art semi-supervised partition methods, which are suffering from less network structure information.
Expansion of benchmark nodes are performed by \eqref{eq:sigma} with the benchmark confidence measure \eqref{eq:confidence}.
New benchmark nodes are selected as shown in Fig. \ref{Toy}(c), where four nodes $v_{6,7,30,34}$ are inserted into two respective benchmark sets.
The followed partition procedure through ALM-based algorithm significantly improves the accuracy of community partition by $57.13\%$, see Fig. \ref{Toy}(d)!

\subsubsection{Optimization Stages}
\begin{figure}[!h]
  \includegraphics[width=0.5\textwidth,height=6.5cm]{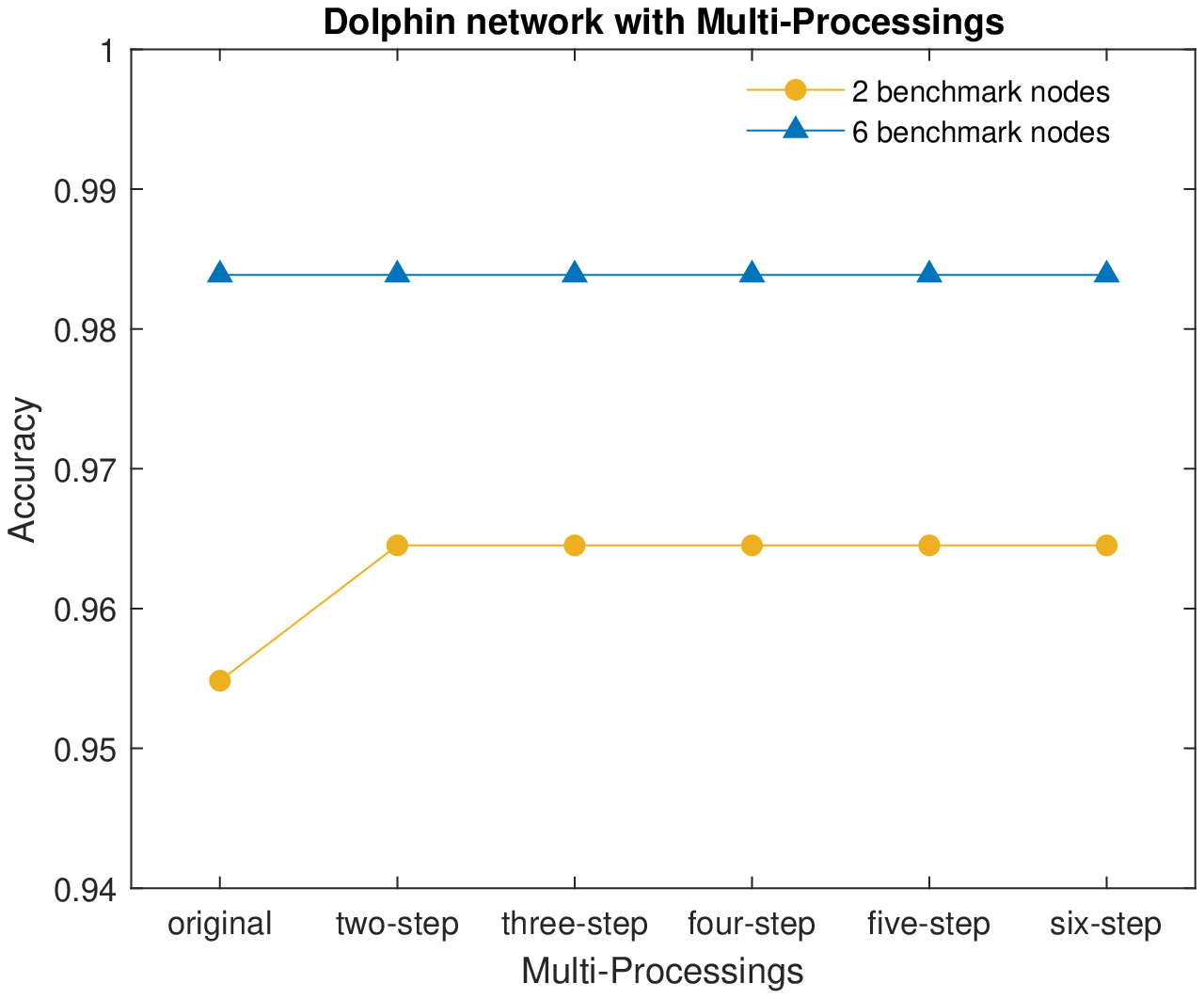}
  \includegraphics[width=0.5\textwidth,height=6.5cm]{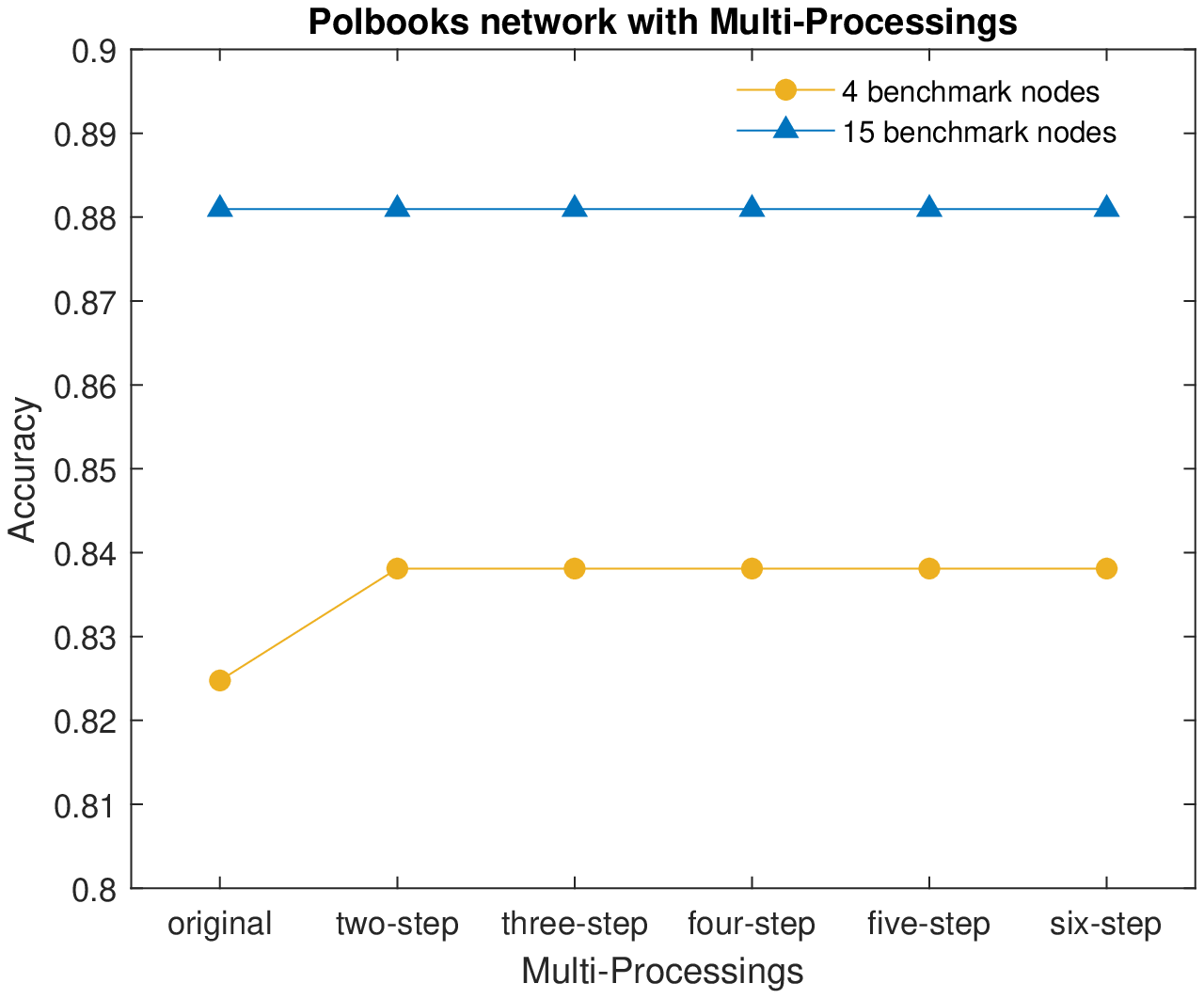}
  \includegraphics[width=0.5\textwidth,height=6.5cm]{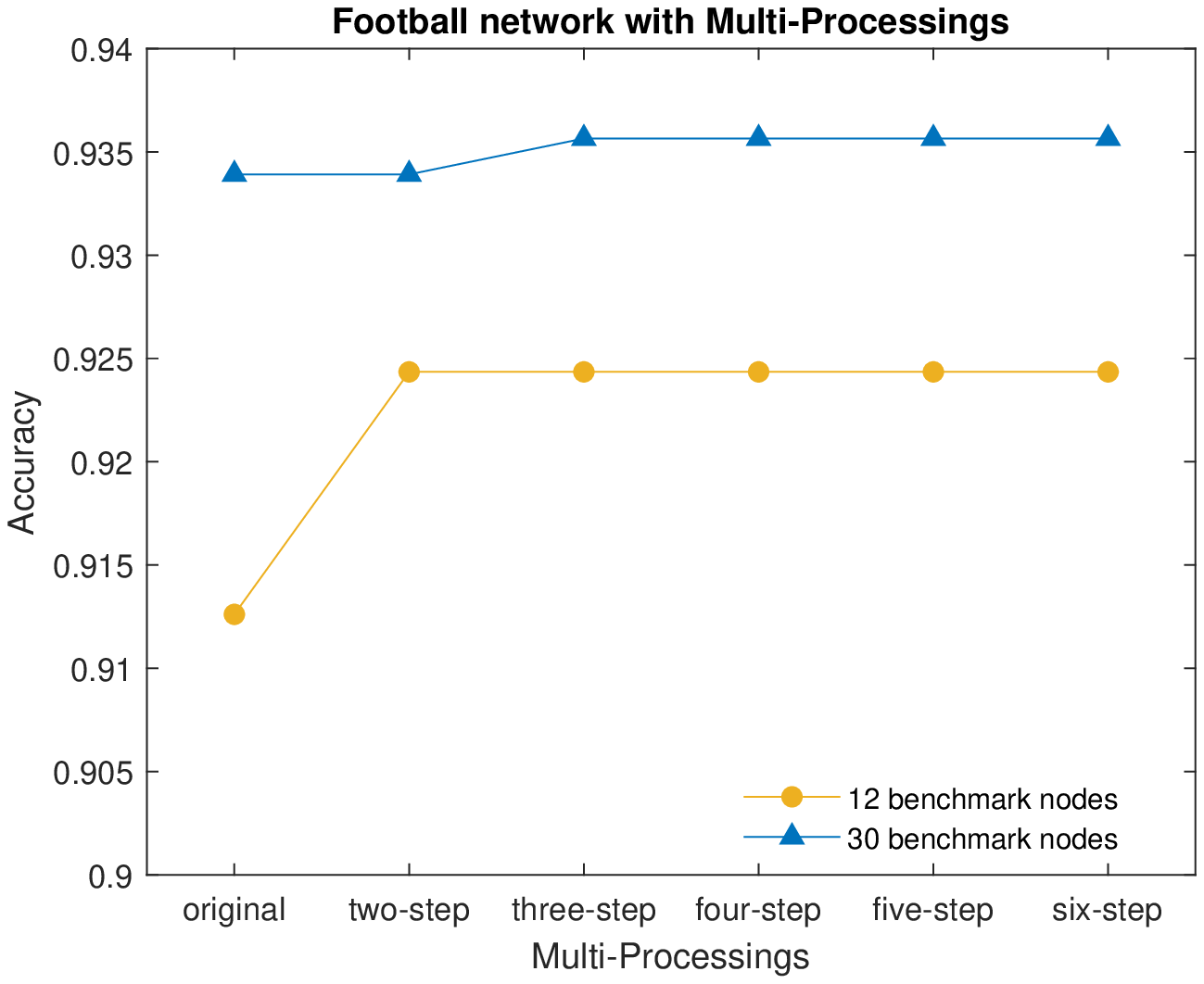}
  \includegraphics[width=0.5\textwidth,height=6.5cm]{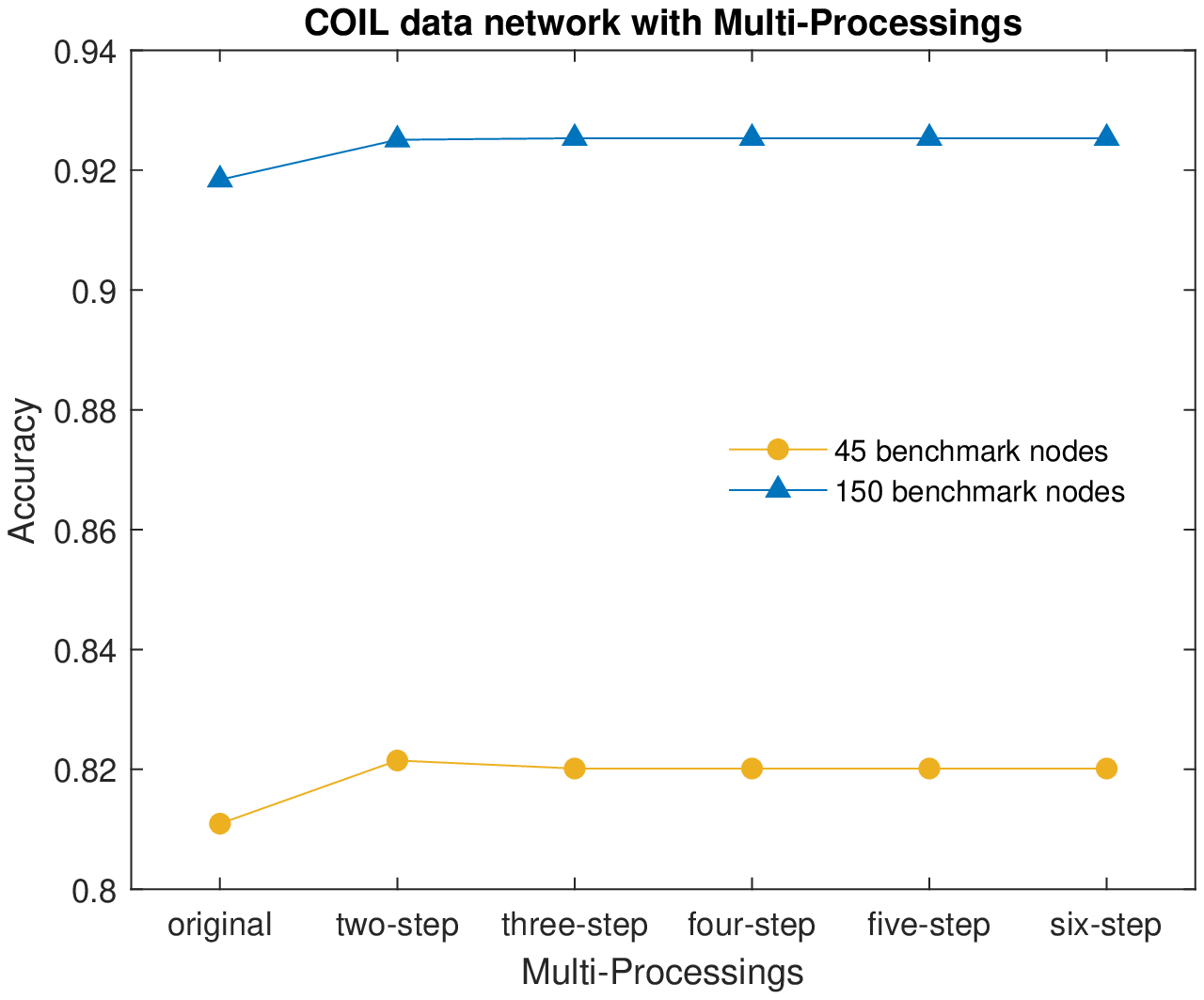}
\caption{\label{exp:ssl} Experiments over four networks of Dolphin network, Polbooks network, Football network and COIL, using more than two optimization stages: the experiment results, with different numbers of benchmark nodes (yellow and blue), show that taking more than two optimization stages does not essentially improve the final accuracy of network partition.}
\end{figure}

The procedures of optimization and benchmark expansion, as Alg. \ref{alg-1}, can be performed not only two but also more than two times, i.e. with multiple optimization stages. Experiment results shown in Fig. \ref{exp:ssl} indicate that the proposed two-stage-optimization strategy can reach the result with enough accuracy, performing more than two optimization stages does not improve the results significantly.

\subsubsection{Selection of Parameter $\delta$ in \eqref{eq:sigma}}
\begin{figure*}[!h]
\centering
\includegraphics[width=0.66\textwidth]{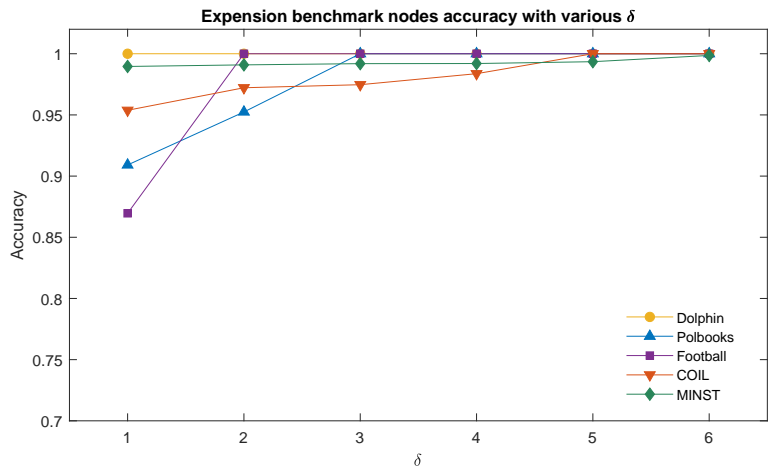}
\caption{\label{exp:view}  Experiment results with different $\delta$ for choosing new benchmark nodes, over the five networks used in this study, including Dolphin network, Polbooks network, Football network, COIL and MINST; clearly, most experiment results do not get noticeable improvements when $\delta \geq 3$. In this work, $\delta$ for experiments over the datasets of MINST and COIL are set $5$, and for the other networks we set $\delta = 3$.
}
\end{figure*}

Nodes with high `confidence' values are the good options for benchmarks. In order to ensure each new benchmark node chosen correctly, the value of $\delta$ should be selected high enough.
By Chebyshev's inequality \cite{Ghahramani2000Fundamentals}, it confirmed that, for any distribution, the amount of data within $\delta$ times of standard deviations is at least ratio $1-{\frac{1}{\delta^2}}$, which means
\bq
{Pr(\left |\pi - \bar{\pi}  \right | \geq \delta \sigma)} \leq {\frac{1}{\delta^2}}
\eq
Thus, the high value of $\delta$ with small appearing probability stands out for a 'trustable' selection. 
Experiment results over five different networks, see Fig. \ref{exp:view}, show that most experiments do not get noticeable improvements when $\delta \geq 3$.
In this work, we choose $\delta = 3$ for networks of average degree $\bar{d} \geq 5$; $\delta = 2 $ for networks with average degree $3 \leq \bar{d} < 5$; $\delta = 1$ for networks with average degree $\bar{d} < 3$, for example the graphs of MINST and COIL.

\subsection{Experiments on Artificial Networks of GN and LFR}

\begin{table}[!h]
\caption{{Partition accuracy with various types of GN networks}}
\label{table:GN}\centering
\scalebox{1}{
\begin{tabular}{p{2cm}p{3cm}p{3cm}p{3cm}}
\toprule
Algorithms & Accuracy (\%) \hfill\break Classical & Accuracy (\%) \hfill\break $Z_{out}=3$ & Accuracy (\%) \hfill\break $Z_{out}=6$ \\
\midrule
TVRF(3\%) &$100(\pm 0)$ &$97.67(\pm1.79)$ &$77.66(\pm 5.23)$ \\
TSOS(3\%) &$100(\pm 0)$ &$\mathbf{100(\pm 0)}$ & $\mathbf{87.18(\pm 6.08)} $\\
\bottomrule
TVRF(6\%) &$100(\pm 0)$ &$98.96(\pm 0.96)$ &$81.42(\pm 3.93)$\\
TSOS(6\%) &$100(\pm 0)$ &$\mathbf{100(\pm 0)}$ &$\mathbf{95.5(\pm 2.53)}$ \\
\bottomrule
\end{tabular}}
\end{table}

GN artificial network \cite{newman2004finding} is proposed by Grivan and Newman, which is still one most popular topics discussed in related literatures. It provides a basic node set \cite{fortunato2016community} with $K=4$ communities, the average total degree of each node is fixed to $16$. At the same time, GN provides a flexible network generation mechanism which is controlled by the number of nodes of each community $n_k$, the number of communities $K$, the number of internal half-edges per node $Z_{in}$, and the number of external half-edges per node $Z_{out}$ etc. Studies \cite{guimera2005functional} show that the parameters $Z_{in}$ and $Z_{out}$ determine the detectable of network communities. Higher value of $Z_{out}$ decreases the detectability of network communities \cite{krause2003compartments}.
In this study, we test our proposed TSOS comparing with TVRF, over three different types of GN networks including the classical GN network and its two variants with different $Z_{out}$ ($Z_{out}=3$ and $Z_{out}=6$), for which each community has at least one benchmark node and the fraction of benchmark nodes is set as $3\%$ and $6\%$.

As the results shown in Table \ref{table:GN}, picking more benchmark nodes results in higher partition accuracy while the same algorithm configuration is set up. For the classical GN network, both algorithms can reach $100\%$ accuracy when only $3\%$ nodes are used as benchmark.  The proposed TSOS can still obtain a completely correct result for the GN network with $Z_{out} = 3$, and a much higher partition accuracy than TVRF for the difficult case with $Z_{out} = 6$. This shows the effectiveness of the proposed strategy by incorporating new benchmark nodes into an additional step of network partition refinement.

\begin{table*}[!t]
\caption{{Partition accuracy with various types of LFR networks}}
\label{table:LFR}\centering
\scalebox{1}{
\begin{tabular}{lp{2.3cm}p{2.3cm}p{2.3cm}p{2.3cm}p{2.3cm}}
\toprule
Algorithms & Accuracy (\%) \hfill\break $\mu=0.1$ & Accuracy (\%) \hfill\break $\mu=0.2$ & Accuracy (\%) \hfill\break $\mu=0.3$ & Accuracy(\%) \hfill\break $\mu=0.4$  & Accuracy(\%) \hfill\break $\mu=0.5$\\
\midrule
TVRF(4\%)  & $97.77(\pm0.519)$       &  $94.37(\pm0.766)$  &  $89.06(\pm0.854)$                & $78.9(\pm1.313)$               & $65.2(\pm3.00)$                    \\
TSOS(4\%) & $\mathbf{99.92(\pm0.021)}$    &  $\mathbf{99.26(\pm0.375)}$ &   $\mathbf{96.38(\pm0.884)}$              & $\mathbf{87.06(\pm2.300)}$ &  $\mathbf{75.72(\pm2.29)}$\\
\bottomrule
TVRF(8\%) &  $98.75(\pm0.357)$           &  $94.86(\pm0.726)$ &   $89.77(\pm1.0371)$                       & $83.43(\pm1.3929)$    &  $71.86(\pm1.59)$\\
TSOS(8\%) &  $\mathbf{100(\pm0)}$     &  $\mathbf{99.8(\pm0.133)}$  & $\mathbf{98.08(\pm0.54)}$        & $\mathbf{93.58(\pm0.801)} $  & $\mathbf{83.18(\pm1.40)}$   \\
\bottomrule
\end{tabular}}
\end{table*}

In contrast to the homogeneous GN networks whose nodes have the same degree, which is actually
not a good proxy of real networks with community structure, the artificial benchmark LFR network, proposed by Lancichinetti, Fortunato and Radicchi \cite{lancichinetti2008benchmark}, has a power law distribution of degree.
LFR benchmark is basically a configuration model with built-in communities \cite{bollobas2004extremal}, which is built by joining stubs at random selection, once one has established which stubs are internal and which ones are external to the stubs \cite{fortunato2016community}. The mixing parameters $\mu_{i}$ is the ratio between the external degree $ext$ and the degree $d_i$ of each vertex $i$ i.e. ${\mu _i} = d_i^{ext}/{d_i}$. Obviously, when $\mu $ is low, each community can be better separated from the others. Here, we generate $5$ networks including $n=1000$ nodes, with the value of $\mu$ ranging from [0.1 0.5], the distributions of degree $d$ and community size $\abs{C}$ follow respective power laws of $d^{-2}$ and $\abs{C}^{-1}$, the average degree is set to $15$ and the community sizes $\abs{C_k}$, $k=1 ... K$, are set from $20$ to $50$.

In the experiments, each community has at least one benchmark node; $4\%$ and $8\%$ nodes are selected as benchmarks for each experiment, so about $40$ and $80$ benchmark nodes are picked, which are slightly bigger than the total number of communities, i.e. rather small samples. As shown in Tab.\ref{table:LFR}, when $\mu$ increases, our proposed TSOS method can still keep the results with high accuracy and perform much better than the TVRF algorithm, hence more robust to increasing external degree, i.e. high mixing parameter $\mu$ does not affect the performance of TSOS more than TVRF. On the other hand, choosing more benchmark nodes promotes both algorithms' performance; however, the proposed TSOS gets improved more significantly.

\subsection{Experiments on Real-World Networks}

In this work, five real-world networks are used to validate the proposed TSOS method, which includes three classical networks of Dolphin network \cite{Lusseau2003The}, Football network \cite{girvan2002community} and Political book network \cite{Newman2006Modularity}, and two data clustering sets of MINST \cite{yin2018effective} and COIL \cite{Chapelle2009Semi}. The three classical social networks are widely used in many community detection studies; the COIL-100 (Columbia object image library-100) data set \cite{Nayar1996Columbia} contains many color images of $100$ different objects, its related graph network used in this paper includes $24$ randomly selected objects (1500 images) from the dataset and the edge weights of the built-up $5$-NN graph are calculated through the Euclidean distance between two images; the MINST data \cite{L1998Gradient} totally consists of 70000 size-normalized and centered images of handwritten digits $0-9$, the images are naturally partitioned to $10$ roughly balanced clusters,  a $10$-NN graph is constructed from the original MINST data set and its edge weights are computed from the Euclidean distance between two images as $784$-dim vectors \cite{yin2018effective}. These networks are considered as undirected and their network parameters are shown in Tab. \ref{table:topology}.

\begin{table}[!h]

\caption{{Five real-world networks are used for experiments with $\bar{d}$ as the average network degree.}}
\label{table:topology}\centering
\scalebox{1}{
\begin{tabular}{lp{2cm}p{2cm}p{2cm}p{1.5cm}p{1.5cm}p{3.2cm}}
\toprule
P & Network & Nodes & Edges & Clusters & $\bar{d}$ & Clustering coefficient  \\
\midrule
1& Dolphin & 62 & 159 & 2& 5.129 & 0.303 \\
2& Polbooks & 105 & 441 & 3& 8.400& 0.488 \\
3& Football & 115 & 613 & 12& 10.66& 0.403 \\
4& COIL& 1500& 3750 &6&5&-\\
5& MNIST & 70000& 350000&10&10&-\\
\bottomrule
\end{tabular}}
\end{table}

\begin{table}[!h]
\caption{Experiments over $5$ real-world networks with various benchmark sizes. Results are averaged under 20 independent trials.\label{table:Real_N_result}}
\centering
\scalebox{1}{
\begin{tabular}{lp{2cm}p{3cm}p{3cm}p{3cm}}
\toprule
P & Network & Benchmark nodes & TVRF (\%) & TSOS (\%) \\
\midrule
1& Dolphin & 2 ($3.2\%$) & $95.43(\pm 4.44)$& $\mathbf{96.29(\pm3.06)}$ \\
2& Dolphin & 6 ($9.7\%$) & $98.38(\pm 0.88)$& $98.38(\pm0.88)$ \\
3& Polbooks & 4 ($3.8\%$) & $81.33(\pm 0.41)$ & $\mathbf{82.86(\pm0.41)}$\\
4& Polbooks &15($14.3\%$) & $88.57(\pm 0.32)$ & $88.57(\pm0.32)$\\
5& Football &12 ($10.7\%$) & $91.17(\pm0.21)$ & $\mathbf{92.14(\pm0.62)}$ \\
6& COIL &45 (3\%)&$80.3(\pm 5.70)$ & $\mathbf{81.38(\pm 6.76)}$ \\
7& COIL &150 (10\%)&$91.7(\pm 2.70)$ & $\mathbf{92.6(\pm 1.93)}$ \\
8& MNIST &70 (0.1\%) & $32.16(\pm 7.82)$& $\mathbf{93.69(\pm3.39)}$\\
9& MNIST &140 (0.2\%) & $89.76(\pm 3.67)$& $\mathbf{97.29(\pm0.13)}$\\
\bottomrule
\end{tabular}}
\end{table}

Experiment results of $5$ real-world networks are illustrated in Tab. \ref{table:Real_N_result}. Similar as the other experiments, picking more benchmark nodes clearly improves network partition accuracy. In addition, the proposed TSOS method performs better for the cases with less initial benchmark nodes, while it can still obtain similar partition accuracy as TVRF for the cases with more initial benchmark nodes. Clearly, for a really small ratio of selected benchmark nodes to the total number of network nodes, e.g. MNIST, TSOS achieves much better partition accuracy than TVRF: $93.69\%$ by TSOS versus $32.16\%$ by TVRF (with $0.1\%$ nodes as benchmark), $97.29\%$ by TSOS versus $89.76\%$ by TVRF (with $0.2\%$ nodes as benchmark). This should thank to the introduced intermediate step of benchmark expansion with a proper confidence criterion.

\begin{figure*}[!h]
\centering
\subfigure[]{
  \includegraphics[width=0.48\textwidth]{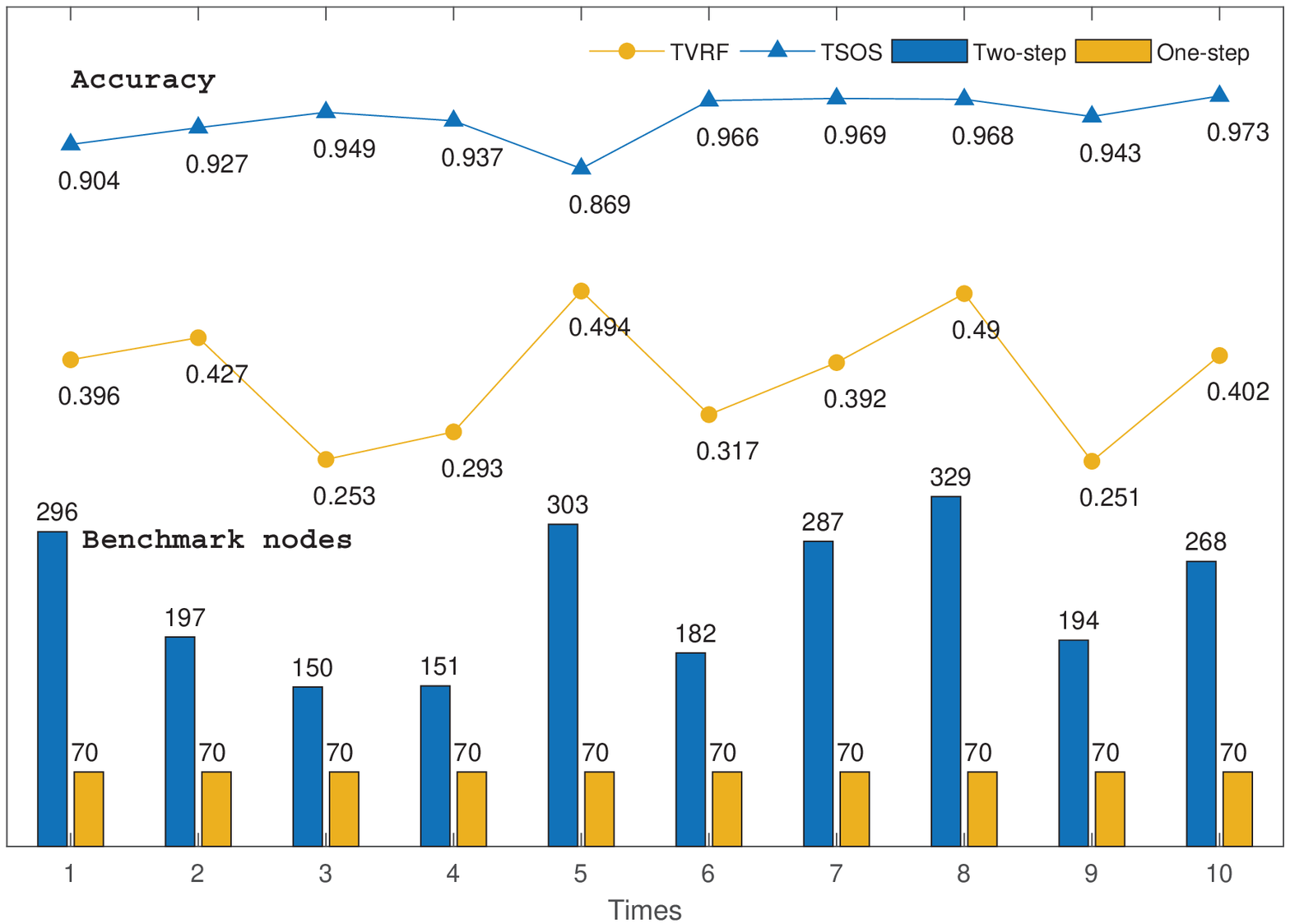}}
  \subfigure[]{\includegraphics[width=0.48\textwidth]{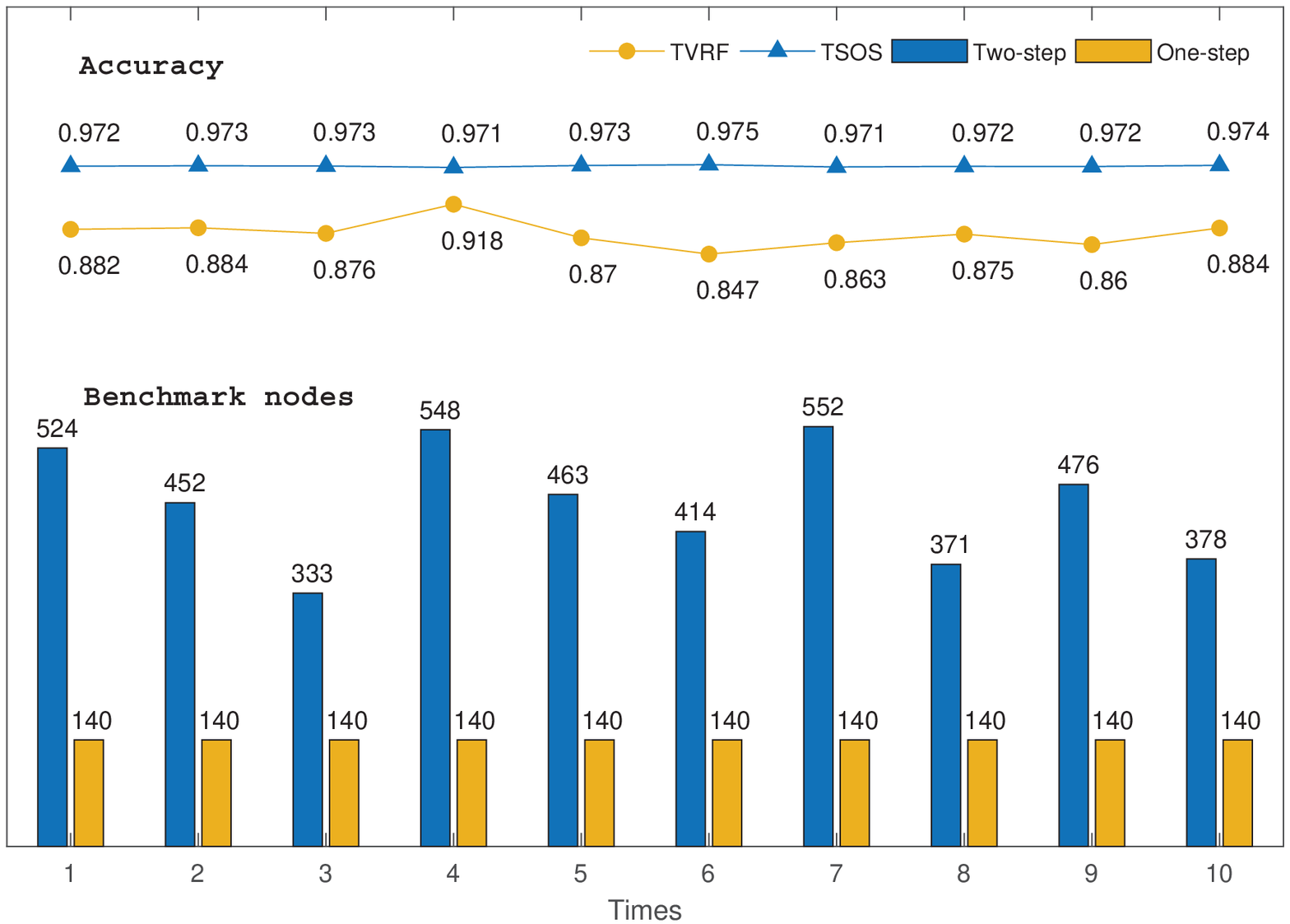}}
\caption{Experiment results on the MNIST dataset: (a). experiments repeat $10$ times, when $0.1\%$ nodes are chosen into benchmark set, TSOS (blue curve) performs more reliable comparing with TVRF (yellow curve), meanwhile, blue bars show the total numbers of benchmark nodes after expansion, which are much more the initial benchmark nodes (yellow bars). (b). also, experiments repeat $10$ times when $0.2\%$ nodes are chosen into benchmark set, TSOS (blue curve) still performs more reliable comparing with TVRF (yellow curve), meanwhile, blue bars show the total numbers of expanded benchmark nodes, which are much more the initial benchmark nodes (yellow bars). }
\label{figure:MIN2}
\end{figure*}

Particularly, for each network partition, we repeat experiments $20$ times with different initial conditions. In view of Tab. \ref{table:Real_N_result} an d Fig. \ref{figure:MIN2}, the computed results through the proposed TSOS often have less variance while keeping higher accuracy, hence better robustness in numerics. For example, clustering MINST data graph with only $0.2\%$ nodes as benchmark, the variance of $20$ experiment results by TSOS is only $0.13\%$, which is much less than the results' variance $3.67\%$ by TVRF. Detailed performance for each experiment setting can be found in Fig. \ref{figure:MIN2}. Moreover, the total numbers of benchmark nodes after expansion are much more the initial benchmark nodes, as blue bars vs. yellow bars shown in Fig. \ref{figure:MIN2}. Such computational robustness is often the seminal factor of partitioning large-scale networks, especially when only a small portion of nodes are available as benchmark.

\section{Conclusions and Future Studies}
We introduce a novel two-stage optimization strategy for partitioning network communities, which makes use of inherent network structure information, i.e. the new network centrality measure of both links and vertices, so as to construct the key affinity description of the given network for which the direct similarities between graph nodes or nodal features are not available to obtain the classical affinity matrix. Such calculated network centrality information presents an essential measure for detecting network communities, and also a `confidence' criterion for developing new benchmark nodes.
We also develop an efficient convex optimization algorithm under the new variational perspective of primal and dual to tackle the challenging combinatorial optimization problem of network partitioning. Experiment results demonstrate that the proposed optimization approach largely improves the accuracy of clustering communities from various networks.

It is obvious that obtaining a reasonable affinity matrix $(w_{ij})$ is the key factor for most graph or network partition algorithms. One way to improve
the effectiveness of the affinity matrix is to take into account the pairs of nodes that are not directly connected, for example, the affinity matrix $\tilde{W}$ through the principle of three degree influence \cite{Fowler2009Dynamic}:
\[
\tilde{W} \, = \, W \, + \, W^2 \, + \, \beta W^3
\]
where $\beta > 0$, or the more generalized affinity matrix $W^*$ given as:
      \[
        {W^ * } = \alpha W + \alpha {W^2} + \alpha {W^3} +  \cdots  = {(I - \alpha W)^{ - 1}} - I
      \]
where $\alpha$ is the attenuation constant which should be less than ${\lambda _{\max }}^{ - 1}(W)$ for convergence.

The computation of each $q^k(e_{ij})$, $r_{i}^k$ and $r_i^s$, for any $e_{ij} \in E$, $k=1 ... K$ and $i=1 ... n$, in the introduced ALM-based dual optimization algorithm (Alg. \ref{alg-2}) can be implemented edgewise and nodewise at the same time, which forms the basis to reimplement the algorithmic steps on modern parallel computing platforms like GPUs or HPCs, so as to significantly improve numerical efficiency and handle super large-scale network partition problems.

\section*{Acknowledgement}
This work was supported by the National Natural Science Foundation of China (Grant Nos. 61877046, 61877047 and 11801200), Shaanxi Provincial Natural Science Foundation of China (Grant No.2017JM1001), the Fundamental Research Funds for the Central Universities, and the Innovation Fund of Xidian University.

\appendices

\section*{Appendix}

\subsection{Equivalent Convex Optimization Models}
\label{sec:app1}

By simple convex analysis, we can equally express the absolute function $w\abs{u}$ as $\max_{q} q \cdot u$, subject to $\abs{q} \leq w$. In this sense, we have the following equivalent expression for each absolute function term of \eqref{eq:potts2}:
\bq \label{eq:app-1}
w_{ij}\abs{\psi _{ik} - \psi _{jk}} \, \Longleftrightarrow \, \max_{q^k(e_{ij})} \, q^k(e_{ij}) (\psi _{ik} - \psi _{jk}) \, , \quad \text{s.t.}\; \abs{q^k(e_{ij})} \, \leq \, w_{ij} \, .
\eq

We can also reformulate the energy term $M_{ik} \psi_{ik}$ of \eqref{eq:potts2} along with the constraint $\psi_{ik} \geq 0$ such that
\bq \label{eq:app-2}
M_{ik} \psi_{ik} \, , \; \psi_{ik} \, \geq\, 0 \, \Longleftrightarrow \, \max_{r_{i}^k}  \, r_i^k \psi_{ik} \, , \quad \text{s.t.} \; r_i^k \, \leq \, M_{ik}\, .
\eq
This is clear that for any $\psi_{ik} < 0$, the maximum of $r_i^k \psi_{ik}$ reaches infinity when $r_i^k$ tends to $- \infty$; for any $\psi_{ik} \geq 0$,
its maximum reaches $M_{ik}\psi_{ik}$ when $r_i^k = M_{ik}$.

In addition, the linear equality constraint of \eqref{eq:simp} can be identically rewritten as
\bq \label{eq:app-3}
\max_{r_i^s} \, r_i^s \big( 1 \, - \, \sum_{k=1}^K \, \psi_{ik} \big) \, ,
\eq
and each variable $r_i^s$ is free.

Observe the facts \eqref{eq:app-1}, \eqref{eq:app-2} and \eqref{eq:app-3}, it is easy to prove that the node-wise simplex constrained convex optimization problem
\eqref{eq:potts2} is mathematically equivalent to the following minimax formulation
\bq \label{eq:pd}
\min_{\psi} \max_{q, r}\;  \sum_{i=1}^n r_i^s \, + \, \sum_{k=1}^K \sum_{i=1}^n \psi_{ik}
\Big( \mathrm{div}(q^k)_i \, - \, r_i^s \, + \, r_i^k \Big) \, , \quad \text{s.t.} \;
\abs{q^k(e_{ij})} \, \leq \, w_{ij} \, , \;\;  r_i^k \, \leq \, M_{ik}\, .
\eq
where the divergence operator $\mathrm{div}(q^k)$ is given in \eqref{eq:div}. In this work,
we call the above optimization problem as the equivalent \emph{primal-dual model}.

While minimizing the primal-dual formulation \eqref{eq:pd} over all $\psi_{ik}$, we can easily obtain the following
maximization problem
\bq \label{eq:app-dual}
 \max_{q, r}\;  \sum_{i=1}^n r_i^s \, , \quad \text{s.t.} \; \;\;
 \mathrm{div}(q^k)_i \, - \, r_i^s \, + \, r_i^k \, = \, 0 \,, \;\;
\abs{q^k(e_{ij})} \, \leq \, w_{ij} \, , \;\;  r_i^k \, \leq \, M_{ik}\, .
\eq

Clearly, the optimization formulation \eqref{eq:app-dual} is also equivalent to the convex optimization problem \eqref{eq:potts2}, which is
named as the equivalent \emph{dual model} in this paper. We actually focus on the optimum $\psi_{ik}$, $i=1 ... n$ and $k=1 ... K$, to the optimization problem \eqref{eq:potts2}, which are the optimal multipliers to the linear equality constraints
\bq \label{eq:appd-lq}
\mathrm{div}(q^k)_i \, - \, r_i^s \, + \, r_i^k \, = \, 0 \, , \quad i\, =\, 1 ... n\, , \;\; k\, =\, 1 ... K\, ,
\eq
in the sense of optimizing its identical dual model \eqref{eq:app-dual}.

\subsection{Detailed Augmented Lagrangian Method Based Algorithm to \eqref{eq:potts2}}
\label{sec:app2}

Details of the proposed augmented Lagrangian method-based algorithm to the linear equality constrained convex optimization problem \eqref{eq:potts2} is listed in Alg. \ref{alg-1}.

\begin{algorithm}[H]
    \caption{Augmented Lagrangian Method Based Algorithm \label{alg-1}}
    \begin{algorithmic}[1]
    \State Choose the proper initial values $(\psi_{ik})^0$, $(q^k(e_{ij}))^0$, $(r_i^s)^0$ and $(r_i^k)^0$ and let $t=1$, start iterations till converged:
    \While{ "not  converged"}
    \State We first compute the residue $(R_i^k)^{t-1}$ at each node $v_i$ and $k$, where $i=1 ... n$ and $k=1 ... K$:
    \[
    (R_{i}^k)^{t-1} \, = \, \big(\mathrm{div}(q^k)_i \, - \, r_i^s \, + \, r_i^k\big)^{t-1} \, - \,  (\psi_{ik})^{t-1}/c \, .
    \]
    \State Fix the values of $(\psi_{ik})^{t-1}$, $(r_i^s)^{t-1}$ and $(r_i^k)^{t-1}$, compute $(q^k(e_{ij}))^{t}$ for each edge $e_{ij} \in E$ and $k=1 ... K$: 
    \[
    (q^k(e_{ij}))^{t} \, = \, \textbf{Projection}_{\abs{q^k(e_{ij})} \leq w_{ij}} \Big((q(e_{ij})^k)^{t-1} \, - \, s \nabla_{e_{ij}} (R^k)^{t-1} \Big)
    \]
    where the projection operator is to threshold the value within the bound $[- w_{ij}, w_{ij}]$, and $s>0$ is the chosen step-size for gradient descent.
    \State Fix the values of $(\psi_{ik})^{t-1}$, $(r_i^s)^{t-1}$ and $(q^k(e_{ij}))^{t}$,  compute $(r_i^k)^{t}$, $i=1 ... n$ and $k=1 ... K$:
    \[
    (r_i^k)^{t} \, = \, \textbf{Projection}_{r_i^k \leq M_{ik}}  \Big((\psi_{ik})^{t-1}/c \, +\, (r_i^s)^{t-1} \, - \, \mathrm{div}((q^k)^t)_i \Big)
    \]
    where the projection operator is to threshold the computation result below the given upper bound.
    \State Fix the values of $(\psi_{ik})^{t-1}$, $(r_i^k)^{t}$ and $(q^k(e_{ij}))^{t}$,  compute $(r_i^s)^{t}$, $i=1 ... n$, by maximizing $L_c(\psi, q, r)$ over each $r_i^s$, which results in
    \[
    (r_i^s)^t \, = \, \Big(1 \, + \, c \sum_{k=1}^K Q_i^k \Big) \, / \, ( c K) \, , \quad Q_i^k \, = \,  \big(\mathrm{div}(q^k)_i \, + \, r_i^k\big)^{t} \, - \,  (\psi_{ik})^{t-1}/c \, .
    \] 
    \State {\bf{update}} $(\psi)_{ik}^t$ as follows:
    \[
    (\psi_{ik})^t \, = \, (\psi_{ik})^{t-1} - c \big(\mathrm{div}(q^k)_i \, - \, r_i^s \, + \, r_i^k\big)^{t}\, , \quad i\, = \, 1 \dots n\, , \;\; k \, = \, 1 \ldots K\, .
    \]
    \State {\bf{update}} $t=t+1$
    \EndWhile
    \State $\mathbf{end while}$
    \State $\mathbf{return}$.
    \end{algorithmic}
\end{algorithm}

\ifCLASSOPTIONcaptionsoff
  \newpage
\fi



%
%
\bibliographystyle{IEEEtran}
\bibliography{Reference}

\begin{thebibliography}{10}
\providecommand{\url}[1]{#1}
\csname url@samestyle\endcsname
\providecommand{\newblock}{\relax}
\providecommand{\bibinfo}[2]{#2}
\providecommand{\BIBentrySTDinterwordspacing}{\spaceskip=0pt\relax}
\providecommand{\BIBentryALTinterwordstretchfactor}{4}
\providecommand{\BIBentryALTinterwordspacing}{\spaceskip=\fontdimen2\font plus
\BIBentryALTinterwordstretchfactor\fontdimen3\font minus
  \fontdimen4\font\relax}
\providecommand{\BIBforeignlanguage}[2]{{%
\expandafter\ifx\csname l@#1\endcsname\relax
\typeout{** WARNING: IEEEtran.bst: No hyphenation pattern has been}%
\typeout{** loaded for the language `#1'. Using the pattern for}%
\typeout{** the default language instead.}%
\else
\language=\csname l@#1\endcsname
\fi
#2}}
\providecommand{\BIBdecl}{\relax}
\BIBdecl

\bibitem{Guanrong2014Network}
Guanrong(Ron)Chen, ``Network science research: some recent progress in china
  and beyond,'' \emph{National Science Review}, vol.~1, no.~3, p. 334, 2014.

\bibitem{Watts1998Collectivedynamics}
W.~DJ and S.~SH, ``Collective dynamics of small-world networks,''
  \emph{Nature}, pp. 440--442, 1998.

\bibitem{Barabasi1999Albert}
A.~L. Barabasi and R.~Albert, ``Emergence of scaling in random networks,''
  \emph{Science}, vol. 286, no. 5439, pp. 509--512, 1999.

\bibitem{mej2010networks}
M.~E. Newman, \emph{Networks: an introduction}.\hskip 1em plus 0.5em minus
  0.4em\relax Oxford University Press, Oxford, 2010.

\bibitem{boccaletti2006complex}
S.~Boccaletti, V.~Latora, Y.~Moreno, M.~Chavez, and D.-U. Hwang, ``Complex
  networks: Structure and dynamics,'' \emph{Physics reports}, vol. 424, no.
  4-5, pp. 175--308, 2006.

\bibitem{mitchell2006complex}
M.~Mitchell, ``Complex systems: Network thinking,'' \emph{Artificial
  Intelligence}, vol. 170, no.~18, pp. 1194--1212, 2006.

\bibitem{goldenberg2001talk}
J.~Goldenberg, B.~Libai, and E.~Muller, ``Talk of the network: A complex
  systems look at the underlying process of word-of-mouth,'' \emph{Marketing
  letters}, vol.~12, no.~3, pp. 211--223, 2001.

\bibitem{palla2005uncovering}
G.~Palla, I.~Der{\'e}nyi, I.~Farkas, and T.~Vicsek, ``Uncovering the
  overlapping community structure of complex networks in nature and society,''
  \emph{Nature}, vol. 435, no. 7043, p. 814, 2005.

\bibitem{fortunato2010community}
S.~Fortunato, ``Community detection in graphs,'' \emph{Physics Reports}, vol.
  486, no. 3-5, pp. 75--174, 2010.

\bibitem{yang2015unified}
L.~Yang, X.~Cao, D.~Jin, X.~Wang, and D.~Meng, ``A unified semi-supervised
  community detection framework using latent space graph regularization,''
  \emph{IEEE Transactions on Cybernetics}, vol.~45, no.~11, pp. 2585--2598,
  2015.

\bibitem{rosvall2007information}
M.~Rosvall and C.~T. Bergstrom, ``An information-theoretic framework for
  resolving community structure in complex networks,'' \emph{Proceedings of the
  National Academy of Sciences}, vol. 104, no.~18, pp. 7327--7331, 2007.

\bibitem{newman2004detecting}
M.~E. Newman, ``Detecting community structure in networks,'' \emph{The European
  Physical Journal B}, vol.~38, no.~2, pp. 321--330, 2004.

\bibitem{girvan2002community}
M.~Girvan and M.~E. Newman, ``Community structure in social and biological
  networks,'' \emph{Proceedings of the National Academy of Sciences of the
  United States of America}, vol.~99, no.~12, pp. 7821--7826, 2002.

\bibitem{guimera2004modularity}
R.~Guimera, M.~Sales-Pardo, and L.~A.~N. Amaral, ``Modularity from fluctuations
  in random graphs and complex networks,'' \emph{Physical Review E}, vol.~70,
  no.~2, p. 025101, 2004.

\bibitem{duch2005community}
J.~Duch and A.~Arenas, ``Community detection in complex networks using extremal
  optimization,'' \emph{Physical Review E}, vol.~72, no.~2, p. 027104, 2005.

\bibitem{liu2014multiobjective}
C.~Liu, J.~Liu, and Z.~Jiang, ``A multiobjective evolutionary algorithm based
  on similarity for community detection from signed social networks,''
  \emph{IEEE Transactions on Cybernetics}, vol.~44, no.~12, pp. 2274--2287,
  2014.

\bibitem{ma2014multi}
L.~Ma, M.~Gong, J.~Liu, Q.~Cai, and L.~Jiao, ``Multi-level learning based
  memetic algorithm for community detection,'' \emph{Applied Soft Computing},
  vol.~19, pp. 121--133, 2014.

\bibitem{Bu2017Dynamic}
Z.~Bu, H.~J. Li, J.~Cao, Z.~Wang, and G.~Gao, ``Dynamic cluster formation game
  for attributed graph clustering,'' \emph{IEEE Transactions on Cybernetics},
  vol.~PP, no.~99, pp. 1--14, 2017.

\bibitem{ali2017improved}
H.~T. Ali and R.~Couillet, ``Improved spectral community detection in large
  heterogeneous networks,'' \emph{The Journal of Machine Learning Research},
  vol.~18, no.~1, pp. 8344--8392, 2017.

\bibitem{Zhu03semi-supervisedlearning}
X.~Zhu, Z.~Ghahramani, and J.~Lafferty, ``Semi-supervised learning using
  gaussian fields and harmonic functions,'' in \emph{ICML}, 2003, pp. 912--919.

\bibitem{Azran07therendezvous}
A.~Azran, ``The rendezvous algorithm: Multiclass semi-supervised learning with
  markov random walks,'' in \emph{ICML}, 2007.

\bibitem{Chung1996Spectral}
F.~Chung, ``Spectral graph theory,'' \emph{CBMS regional conference series in
  mathematics}, no.~92, 1996.

\bibitem{Luxburg2007A}
U.~Luxburg, \emph{A tutorial on spectral clustering}.\hskip 1em plus 0.5em
  minus 0.4em\relax Kluwer Academic Publishers, 2007.

\bibitem{Hagen1991Fast}
L.~Hagen and A.~Kahng, ``Fast spectral methods for ratio cut partitioning and
  clustering,'' in \emph{IEEE International Conference on Computer-Aided
  Design, 1991. Iccad-91. Digest of Technical Papers}, 1991, pp. 10--13.

\bibitem{Shi2000Normalized}
J.~Shi and J.~Malik, ``Normalized cuts and image segmentation,'' \emph{IEEE
  Trans. Pattern Anal. Mach. Intell}, vol.~22, no.~8, pp. 888--905, 2000.

\bibitem{buhler2009spectral}
T.~B{\"u}hler and M.~Hein, ``Spectral clustering based on the graph
  p-laplacian,'' in \emph{Proceedings of the 26th Annual International
  Conference on Machine Learning}.\hskip 1em plus 0.5em minus 0.4em\relax ACM,
  2009, pp. 81--88.

\bibitem{hein2011beyond}
M.~Hein and S.~Setzer, ``Beyond spectral clustering-tight relaxations of
  balanced graph cuts,'' in \emph{Advances in Neural Information Processing
  Systems}, 2011, pp. 2366--2374.

\bibitem{Blum01learningfrom}
A.~Blum and S.~Chawla, ``Learning from labeled and unlabeled data using graph
  mincuts,'' ICML 2001.

\bibitem{Boykov01fastapproximate}
Y.~Boykov, O.~Veksler, and R.~Zabih, ``Fast approximate energy minimization via
  graph cuts,'' \emph{PAMI}, vol.~23, pp. 1222 -- 1239, 2001.

\bibitem{MR2177726}
T.~F. Chan and S.~Esedo{\=g}lu, ``Aspects of total variation regularized {$L\sp
  1$} function approximation,'' \emph{SIAM J. Appl. Math.}, vol.~65, no.~5, pp.
  1817--1837, 2005.

\bibitem{bresson2014multi}
X.~Bresson, X.-C. Tai, T.~F. Chan, and A.~Szlam, ``Multi-class transductive
  learning based on $l_1$ relaxations of cheeger cut and mumford-shah-potts
  model,'' \emph{Journal of Mathematical Imaging and Vision}, vol.~49, no.~1,
  pp. 191--201, 2014.

\bibitem{Yuan2010A}
J.~Yuan, E.~Bae, and X.~C. Tai, ``A study on continuous max-flow and min-cut
  approaches,'' in \emph{Computer Vision and Pattern Recognition}, 2010, pp.
  2217--2224.

\bibitem{yuan2010continuous}
J.~Yuan, E.~Bae, X.-C. Tai, and Y.~Boykov, ``A continuous max-flow approach to
  potts model,'' in \emph{European Conference on Computer Vision}.\hskip 1em
  plus 0.5em minus 0.4em\relax Springer, 2010, pp. 379--392.

\bibitem{yin2018effective}
K.~Yin and X.-C. Tai, ``An effective region force for some variational models
  for learning and clustering,'' \emph{Journal of Scientific Computing},
  vol.~74, no.~1, pp. 175--196, 2018.

\bibitem{Zhou2005Regularization}
D.~Zhou and B.~Sch{\"o}lkopf, ``Regularization on discrete spaces,'' in
  \emph{Joint Pattern Recognition Symposium}.\hskip 1em plus 0.5em minus
  0.4em\relax Springer, 2005, pp. 361--368.

\bibitem{Book2009Zhu}
X.~Zhu and A.~B. Goldberg, \emph{Introduction to Semi-Supervised Learning},
  ser. Synthesis Lectures on Artificial Intelligence and Machine
  Learning.\hskip 1em plus 0.5em minus 0.4em\relax Morgan \& Claypool
  Publishers, 2009.

\bibitem{citeulike1859441}
D.~P. Bertsekas, \emph{Nonlinear Programming}.\hskip 1em plus 0.5em minus
  0.4em\relax {Athena Scientific}, September 1999.

\bibitem{bai2017effective}
Y.~Bai, S.~Liu, and Z.~Zhang, ``Effective hybrid link-adding strategy to
  enhance network transport efficiency for scale-free networks,''
  \emph{International Journal of Modern Physics C}, vol.~28, no.~08, p.
  1750107, 2017.

\bibitem{Dunn2005The}
R.~Dunn, F.~Dudbridge, and C.~M. Sanderson, ``The use of edge-betweenness
  clustering to investigate biological function in protein interaction
  networks,'' \emph{Bmc Bioinformatics}, vol.~6, no.~1, p.~39, 2005.

\bibitem{Bollob1998Modern}
B.~Bollobas, \emph{Modern Graph Theory}, ser. Graduate Texts in
  Mathematics.\hskip 1em plus 0.5em minus 0.4em\relax Springer, 1998, vol. 184.

\bibitem{yang2017small}
Y.~Yang, T.~Nishikawa, and A.~E. Motter, ``Small vulnerable sets determine
  large network cascades in power grids,'' \emph{Science}, vol. 358, no. 6365,
  p. eaan3184, 2017.

\bibitem{dorogovtsev2006k}
S.~N. Dorogovtsev, A.~V. Goltsev, and J.~F.~F. Mendes, ``K-core organization of
  complex networks,'' \emph{Physical Review Letters}, vol.~96, no.~4, p.
  040601, 2006.

\bibitem{Ghahramani2000Fundamentals}
S.~Ghahramani, \emph{Fundamentals of probability}.\hskip 1em plus 0.5em minus
  0.4em\relax Prentice Hall,, 2000.

\bibitem{newman2004finding}
M.~E. Newman and M.~Girvan, ``Finding and evaluating community structure in
  networks,'' \emph{Physical Review E}, vol.~69, no.~2, p. 026113, 2004.

\bibitem{fortunato2016community}
S.~Fortunato and D.~Hric, ``Community detection in networks: A user guide,''
  \emph{Physics Reports}, vol. 659, pp. 1--44, 2016.

\bibitem{guimera2005functional}
R.~Guimera and L.~A.~N. Amaral, ``Functional cartography of complex metabolic
  networks,'' \emph{Nature}, vol. 433, no. 7028, p. 895, 2005.

\bibitem{krause2003compartments}
A.~E. Krause, K.~A. Frank, D.~M. Mason, R.~E. Ulanowicz, and W.~W. Taylor,
  ``Compartments revealed in food-web structure,'' \emph{Nature}, vol. 426, no.
  6964, p. 282, 2003.

\bibitem{lancichinetti2008benchmark}
A.~Lancichinetti, S.~Fortunato, and F.~Radicchi, ``Benchmark graphs for testing
  community detection algorithms,'' \emph{Physical review E}, vol.~78, no.~4,
  p. 046110, 2008.

\bibitem{bollobas2004extremal}
B.~Bollob{\'a}s, \emph{Extremal graph theory}.\hskip 1em plus 0.5em minus
  0.4em\relax Courier Corporation, 2004.

\bibitem{Lusseau2003The}
D.~Lusseau, K.~Schneider, O.~J. Boisseau, P.~Haase, E.~Slooten, and S.~M.
  Dawson, ``The bottlenose dolphin community of doubtful sound features a large
  proportion of long-lasting associations,'' \emph{Behavioral Ecology \&
  Sociobiology}, vol.~54, no.~4, pp. 396--405, 2003.

\bibitem{Newman2006Modularity}
M.~E. Newman, ``Modularity and community structure in networks,'' in \emph{APS
  March Meeting}, 2006, pp. 8577--8582.

\bibitem{Chapelle2009Semi}
O.~Chapelle, B.~Scholkopf, and A.~Z. Eds, ``Semi-supervised learning (chapelle,
  o. et al., eds.; 2006) [book reviews],'' \emph{IEEE Transactions on Neural
  Networks}, vol.~20, no.~3, pp. 542--542, 2009.

\bibitem{Nayar1996Columbia}
S.~Nayar, ``Columbia object image library (coil100),'' 1996.

\bibitem{L1998Gradient}
Y.~Lecun, L.~Bottou, Y.~Bengio, and P.~Haffner, ``Gradient-based learning
  applied to document recognition,'' \emph{Proceedings of the IEEE}, vol.~86,
  no.~11, pp. 2278--2324, 1998.

\bibitem{Fowler2009Dynamic}
J.~H. Fowler and N.~A. Christakis, ``Dynamic spread of happiness in a large
  social network: Longitudinal analysis of the framingham heart study social
  network,'' \emph{BMJ: British Medical Journal}, vol. 338, no. 7685, pp.
  23--27, 2009.

\end{thebibliography}

%

\begin{IEEEbiography}[{\includegraphics[width=1in,height=1.25in,clip,keepaspectratio]{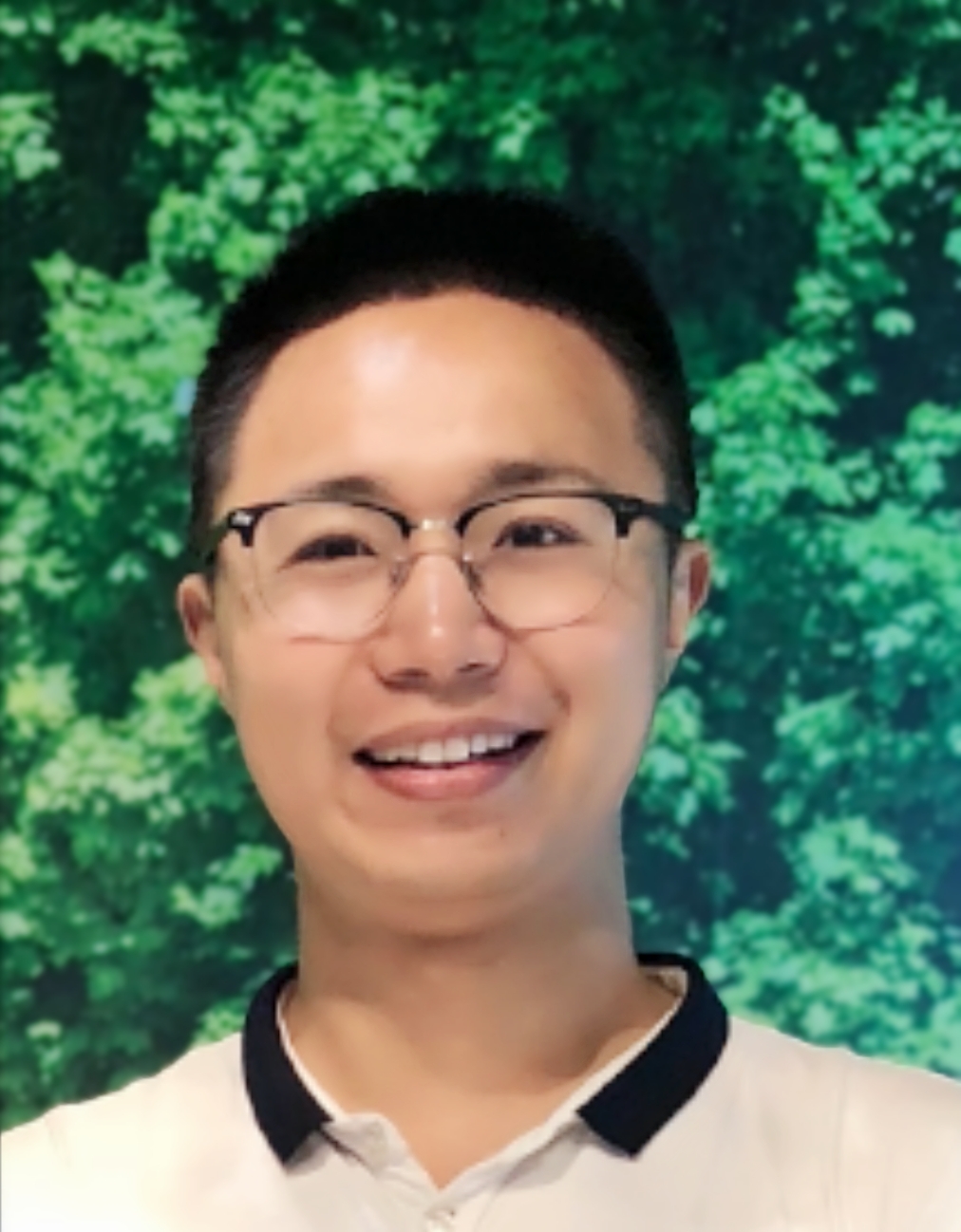}}]{Yiguang Bai}
received the B.S. degree in Mathematics and Applied Mathematics from Xidian University, Xi'an, China, in
2014 and the M.S. degree in Applied Mathematics from Xidian University, Xi'an, China, in 2017

Bai Yiguang is currently pursuing a PhD degree at the School of Mathematics and Statistics, Xidian University. His research interests include network structure, robustness and function.
\end{IEEEbiography}

\begin{IEEEbiography}[{\includegraphics[width=1in,height=1.25in,clip,keepaspectratio]{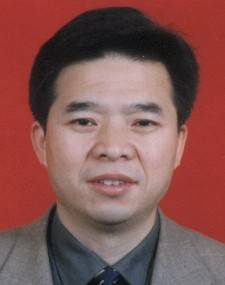}}]{Sanyang Liu}
received the M.S. degree in Applied Mathematics from Xidian
University, and the Ph.D. degree in Computational Mathematics from Xi'an Jiaotong
University, Xi'an in 1984 and 1989 respectively. After finishing the Ph.D. degree, he spent
one year at the University Paul Sabatier, Toulouse, France as a postdoctoral fellow.

He is currently the Director of the Institute of Industrial and Applied Mathematics, Director
of the Center for Mathematics and Interdisciplinary Research in Xidian University. His research interests
include optimization methods and their applications, nonlinear analysis,
system modeling, information network, etc.

Dr. Liu has been selected in the National Special Support Program for High-level Personnel Recruitment, China, 2016.
\end{IEEEbiography}

\begin{IEEEbiography}[{\includegraphics[width=1in,height=1.25in,clip,keepaspectratio]{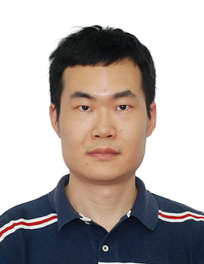}}]{Ke Yin}
is associate professor in Center for Mathematical Sciences at Huazhong University and Science and Technology. He has been an Adjunct Assistant Professor and postdoctoral fellow in Department of Mathematics at UCLA from 2013 to 2016. Before that He obtained BS from University of Science and Technology of China, and PhD from Georgia Institute of Technology in USA. Dr. Yin’s research interests include inverse problems, non-smooth optimization with applications in machine learning and image processing.
\end{IEEEbiography}

\begin{IEEEbiography}[{\includegraphics[width=1in,height=1.25in,clip,keepaspectratio]{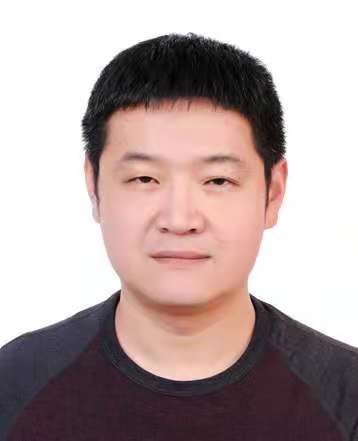}}]{Jing Yuan}
is working as the professor at the School of Mathematics and Statistics, Xidian University in Xi'an, China. Before that, he worked as the research scientist at Robarts Research Institute of Western University in Canada from 2011 to 2016. He obtained his PhD with excellence from the Department of Computer Science and Mathematics in Heidelberg University, Germany. His research interests are in developing convex optimization theories and algorithmic implementations, advanced variational analysis and high-performance distributed parallel computing, especially with applications to most challenging practices of computer vision, medical image analysis and machine learning. He published about 100 papers in the top international journals and conferences, and is serving as the committee member or reviewer of many top conferences and journals.
\end{IEEEbiography}

\end{document}